\documentclass[12pt]{article} 
\usepackage[sectionbib]{natbib}
\usepackage{array,epsfig,fancyheadings,rotating}
\usepackage[]{hyperref}  

\usepackage{amsmath}
\usepackage{amssymb}
\usepackage{amsfonts}
\usepackage{multirow}
\usepackage{amsthm}
\usepackage{float}
\usepackage{color,soul}
\usepackage{booktabs}

\usepackage{fancyvrb}
  
\newtheorem{res}{Result}[section]

\newtheorem{lemma}{Lemma}
\newtheorem{corollary}{Corollary}
\newtheorem{proposition}{Proposition}
\theoremstyle{definition}

\newtheorem{example}{Example}

\begin{document}


\renewcommand{\baselinestretch}{2}

\markright{ \hbox{\footnotesize\rm 
}\hfill\\[-13pt]
\hbox{\footnotesize\rm
}\hfill }

\markboth{\hfill{\footnotesize\rm FIRSTNAME1 LASTNAME1 AND FIRSTNAME2 LASTNAME2} \hfill}
{\hfill {\footnotesize\rm Estimation of MNAR Data with IV} \hfill}

\renewcommand{\thefootnote}{}
$\ $\par


\fontsize{12}{14pt plus.8pt minus .6pt}\selectfont \vspace{0.8pc}

\centerline{\large\bf Semiparametric Estimation with Data Missing}
\vspace{2pt} \centerline{\large\bf Not at Random Using an Instrumental Variable}
\vspace{.4cm} \centerline{ BaoLuo Sun{\normalsize $^1$}, Lan Liu{\normalsize $^1$}, Wang Miao{\normalsize $^{1,4}$}, Kathleen Wirth{\normalsize $^{2,3}$},} 
\centerline{James Robins{\normalsize $^{1,2}$} and  Eric J. Tchetgen Tchetgen{\normalsize $^{1,2}$}}
\vspace{.4cm} 

\centerline{\it
Departments of Biostatistics{\normalsize $^1$}, Epidemiology{\normalsize $^2$} and Immunology }

\centerline{\it and Infectious Diseases{\normalsize $^3$}, Harvard T.H. Chan School of Public Health. }

\centerline{\it Beijing International Center for Mathematical Research{\normalsize $^4$}, Peking University.} \vspace{.55cm} \fontsize{9}{11.5pt plus.8pt minus
.6pt}\selectfont


\begin{quotation}
\noindent {\it Abstract:}
Missing data occur frequently in empirical studies in health and social sciences, and can compromise our ability to obtain valid inference. An outcome is said to be missing not at random (MNAR) if, conditional on the observed variables, the missing data mechanism still depends on the unobserved outcome. In such settings, identification is generally not possible without imposing additional assumptions. Identification is sometimes possible, however, if an instrumental variable (IV) is observed for all subjects which satisfies the exclusion restriction that the IV affects the missingness process without directly influencing the outcome. In this paper, we provide necessary and sufficient conditions for nonparametric identification of the full data distribution under MNAR with the aid of an IV. In addition, we give sufficient identification conditions that are more straightforward to verify in practice. For inference, we focus on estimation of a population outcome mean, for which we develop a suite of semiparametric estimators that extend methods previously developed for data missing at random. Specifically, we propose a novel doubly robust estimator of the mean of an outcome subject to MNAR. For illustration, the methods are used to account for selection bias induced by HIV testing refusal in the evaluation of HIV seroprevalence in Mochudi, Botswana, using interviewer characteristics such as gender, age and years of experience as IVs.\par

\vspace{9pt}
\noindent {\it Key words and phrases:}
Instrumental variable, Missing not at random, Inverse probability weighting, Doubly robust.
\par
\end{quotation}\par

\def\thefigure{\arabic{figure}}
\def\thetable{\arabic{table}}

\renewcommand{\theequation}{\thesection.\arabic{equation}}

\fontsize{12}{14pt plus.8pt minus .6pt}\selectfont

\setcounter{section}{1} 
\setcounter{equation}{0} 
\noindent {\bf 1. Introduction}

\lhead[\footnotesize\thepage\fancyplain{}\leftmark]{}\rhead[]{\fancyplain{}\rightmark\footnotesize\thepage}

Selection bias is a major problem in health and social sciences, and is said to be present  in an empirical study if features of the underlying population of primary interest are entangled with features of the selection process not of scientific interest. Selection bias can occur in practice due to incomplete data, if the observed sample is not representative of the underlying population. While various ad hoc methods exist to adjust for missing data, such methods may be subject to bias unless under fairly strong assumptions. For example, complete-case analysis is easy to implement and is routinely used in practice. However, complete-case analysis can be biased when the outcome is not missing completely at random (MCAR) \citep{rubin1}. Progress can still be made if data are missing at random (MAR), such that the missing data mechanism is independent of unobserved variables conditional on observed variables. Principled methods for handling MAR data abound, including likelihood-based procedures \citep{rubin1, horton2}, multiple imputation \citep{rubin2, kenward2, horton,schafer}, inverse probability weighting \citep{jamie2,tsiatis, linglingli} and doubly robust estimation \citep{rotni3,lip,robins,robi,neu,tsiatis,eric}.

The MAR assumption is strictly not testable in a nonparametric model without an additional assumption \citep{dgill, potthoff} and is often untenable. An outcome is said to be missing not at random (MNAR) if it is neither MCAR nor MAR, such that conditional on the observed variables, the missingness process depends on the unobserved variables \citep{rubin1}. Identification is generally not available under MNAR without an additional assumption \citep{rit}. A possible approach is to make sufficient parametric assumptions \citep{rubin1, roy,wu} about the full data distribution for identification. However, this approach can fail even with commonly used fully parametric models \citep{miao2014normal,wang2014instrumental}. Alternative strategies for MNAR include positing instead sufficiently stringent modeling restrictions on a model for the missing data process \citep{rot} or conducting sensitivity analysis and constructing bounds  \citep{betancur,kenward,robins, vans}. {A framework for identification and semiparametric inference was recently proposed by }\citet{miao2015identification} and \citet{miao2016varieties}, {building on earlier work by }\citet{d2010new}, \citet{wang2014instrumental} and \citet{zhao2015semiparametric}, {under the assumption that a shadow variable is fully observed which is associated with the outcome prone to missingness but independent of the missingness process conditional on covariates and the possibly unobserved outcome.}
Another common identification approach involves leveraging an instrumental variable (IV) \citep{man, winship}. Heckman's framework \citep{heckman1,heckman2} is perhaps the most common IV approach used primarily in economics and other social sciences to account for outcome MNAR. A valid IV is known to satisfy the following conditions:
\begin{description}
\item[(i)] the IV is not directly related to the outcome in the underlying population, conditional on a set of fully observed covariates, and
\item[(ii)] the IV is associated with the missingness mechanism conditional on the fully observed covariates.
\end{description}
Therefore a valid IV must predict a person's propensity to have an observed outcome, without directly influencing the outcome. 

In principle, one can use a valid IV to obtain a nonparametric test of the MAR assumption. However access to an IV does not generally point identify the joint distribution of the full data nor its functionals. {Heckman's selection model consists of an outcome model that is associated with the selection process through correlated latent variables included in both models} \citep{heckman1}. It is generally not identifiable without an assumption of bivariate normal latent error in defining the model \citep{wool}. Estimation using Heckman-type selection models may be sensitive to these parametric assumptions \citep{winship, puhani}, although there has been significant work towards relaxing some of the assumptions \citep{man, newey90, das, newey09}. An alternative sufficient identification condition was considered by \citet{erickath} which involves restricting the functional form of the selection bias function due to non-response on the additive, multiplicative or odds ratio scale. However, their approach for estimation is fully parametric and may be sensitive to bias due to model misspecification. Therefore a more robust approach is warranted. 

In this paper, we develop a general framework for nonparametric identification of selection models based on an IV. We describe necessary and sufficient conditions for identifiability of the full data distribution with a valid IV. For inference we focus on estimation of an outcome mean, although the proposed methods are easy to adapt to other functionals. {We develop semiparametric approaches including inverse probability weighting (IPW) and outcome regression (OR) that extend analogous methods previously developed for missing at random (MAR) settings, and introduce a novel doubly robust (DR) estimation approach.} The consistency of each estimator relies on correctly specified models for different parts of the data generating model. {We note that IPW in MNAR via calibration weighting} \citep{deville2000generalized,kott2006using,chang2008using} {has previously been proposed to account for sample nonresponse in survey design settings, and typically requires matching of weighted estimates to population totals for benchmark variables. Besides assuming a correctly specified model for nonresponse, identification in such settings is made possible by availability of known or estimated population totals, an assumption we do not require.} Extensive simulation studies are used to investigate the finite sample properties of proposed estimators. For illustration, the methods are used to account for selection bias induced by HIV testing refusal in the evaluation of HIV seroprevalence in Mochudi, Botswana, using interviewer characteristics including gender, age and years of experience as IVs. All proofs are relegated to a Supplemental Appendix.

\setcounter{section}{2} 
\setcounter{equation}{0} 
\noindent {\bf 2. Notation and Assumptions}

Suppose that one has observed $n$ independent and identically distributed observations {$(X,Y,R,Z)$ with fully observed covariates $X$ and $R$ is the indicator of whether the person's outcome is observed. $Y$ is observed if $R=1$ and $Y=Y^{\ast}$ otherwise, where $Y^{\ast}$ denotes missing outcome value.} The variable $Z$ is a fully observed IV that satisfies assumptions (i) and (ii) formalized below. In the evaluation of HIV prevalence in Mochudi, $X$  includes all demographic and behavioral variables collected for all persons in the sample, while HIV status $Y$ may be missing for individuals who failed to be tested, i.e. with $R=0$. Let  $\tilde{\pi} (X,Z) =\Pr(R=1|X,Z)$ denote the propensity score for the missingness mechanism given $(X,Z)$.  As a valid IV, we will assume that $Z$ satisfies the following assumptions.

\begin{description}
  \item[(IV.1)] Exclusion restriction: $$P_{Y|X,Z}(y|x,z)=P_{Y|X}(y|x)\quad \text{$\forall$ $x,z$.}$$  
  \item[(IV.2)] IV relevance: $$\tilde{\pi} (x,z)  \neq \tilde{\pi} (x,z^{\prime}) \quad \text{$\forall$ $x$.}$$ 
\end{description}  

Exclusion restriction (IV.1) states that the IV and the outcome are conditionally independent given $X$ in the underlying population, that is the IV does not have a direct effect on the outcome, which places restrictions on the full data law. IV relevance requires that the IV remains associated with the missingness mechanism even after conditioning on $X$. In spite of (IV.2), (IV.1) implies that $Z$ cannot reduce the dependence between $R$ and $Y$, therefore under MNAR $\pi (x,y,z)=P(R=1|x,y,z)$ remains a function of $y$ even after conditioning on $(x,z)$. In addition, (IV.1) and (IV.2) imply that under MNAR the IV remains relevant in $\pi (x,y,z)$ conditional on $(x,y)$. Both of these facts will be used repeatedly throughout. $\tilde{\pi} (x,z)$ is typically referred to as the propensity score for the missingness process, and we shall likewise refer to $\pi (x,y,z)$ as the extended propensity score.

\setcounter{section}{3} 
\setcounter{equation}{0} 
\noindent {\bf 3. Identification}

Although (IV.1) reduces the number of unknown parameters in the full data law, identification is still only available for a subset of all possible full data laws. As an illustration, consider the case of binary outcome and IV. For simplicity and without loss of generality, we omit covariates $X$. Assumption (IV.1) implies $P(z,y)=P(y)P(z)$. We are only able to identify the quantities  $P(z,y|R=1)$, $P(z|R=0)$, $P(R=1)$ from the observed data. These quantities are functions of the unknown parameters: $P(Z=1)$, $P(Y=1)$, and $P(R=1|z,y)$. So we have six unknown parameters, but only five available independent equations, one for each identified parameter given above. As a result, the full data law is not identifiable, and $P(Y=1)$ is not identifiable.

The IV model becomes identifiable once one sufficiently restricts the class of models for the joint distribution of $(Z,Y,R)$. Let $\mathcal{P_\theta}(R,Z,Y)$, $\mathcal{P_\eta}(Z)$ and $\mathcal{P_\xi}(Y)$ denote  the collection of such candidates for $P(R=1|z,y)$, $P(z)$ and $P(y)$, respectively. Members of the sets are indexed by parameters $\theta$, $\eta$ and $\xi$, which may be infinite dimensional. The identifiability of the model is determined by the relationship between its members. 

\begin{res}\label{thm:para}
Suppose that Assumption (IV.1) holds, then the joint distribution $P(z,y,r)$ is identifiable if and only if $\mathcal{P_\theta}(R,Z,Y)$ and  $\mathcal{P_\xi}(Y)$ satisfy the following condition: 
for any pair of candidates $$\left\{P_{\theta_1}(R=1|z,y), P_{\xi_1}(y)\right\} \text{ and }\left\{P_{\theta_2}(R=1|z,y), P_{\xi_2}(y)\right\}$$ in the model the following inequality holds:
\begin{align}\label{eq:inequality}
\frac{P_{\theta_1}(R=1|z,y)}{P_{\theta_2}(R=1|z,y)}\neq \frac{P_{\xi_2}(y)}{P_{\xi_1}(y)}
\end{align}
for at least one value of $z$ and $y$.
\end{res}

Result \ref{thm:para} presents a necessary and sufficient condition for identifiability of the joint distribution of the full data, and thus a sufficient condition for identifiability of its functionals. We have the following corollary which provides a more convenient condition to verify.

\begin{corollary}\label{cor:para}
Suppose that Assumption (IV.1) holds, then the joint distribution $P(z,y,r)$ is identifiable if \text{ $\forall$ $\theta_1, \theta_2$ such that $\theta_1 \neq \theta_2$}, the ratio $ P_{\theta_1}(R=1|z,y) / P_{\theta_2}(R=1|z,y)$ is either a constant or varies with $z$.
\end{corollary}

Although Corollary \ref{cor:para} provides a sufficient condition for identification of the joint distribution of the full data, it may be used to establish identifiability in parametric or semi-parametric models which we illustrate in a number of examples. Let $\mathcal{M}_{\text{\tiny{IV}}}$ denote the collection of models with valid IV. 

\begin{example}\label{ex:ign}
Suppose both $Y$ and $Z$ are binary and consider the model $\mathcal{M}_{\text{\tiny{1}}} \cap \mathcal{M}_{\text{\tiny{IV}}}$, where
\begin{align*}
\mathcal{M}_{\text{\tiny{1}}}=\bigl\{ &P(R=1|Z,Y)=\text{expit} \left[\theta_0 + \theta_1 Z + \theta_2 Y + \theta_3 ZY\right]: \\
&(\theta_0,\theta_1,\theta_2,\theta_3) \in \mathbb{R}^4 \bigr\},
\end{align*}
which includes the saturated model, i.e. the nonparametric model. It is shown in the Supplemental Appendix that this model does not satisfy inequality (\ref{eq:inequality}) and therefore the joint distribution of $(Z,Y,R)$ cannot be identified without reducing the dimension of $\theta$. In contrast, Corollary \ref{cor:para} confirms that the smaller model $\mathcal{M}_{\text{\tiny{2}}} \cap \mathcal{M}_{\text{\tiny{IV}}}$ is identified, where 
$$
\mathcal{M}_{\text{\tiny{2}}} =\left\{ P(R=1|Z,Y)=\text{expit}\left[\theta_0 + \theta_1 Z + \theta_2 Y\right] : (\theta_0,\theta_1,\theta_2) \in \mathbb{R}^3 \right\},
$$that is, the IV model becomes identified upon imposing a no-interaction assumption between $Y$ and $Z$ in the logistic model for the extended propensity score. An analogous result holds for possibly continuous $Y$ and $Z$, assuming the following logistic generalized additive model for the extended propensity score.
\end{example}

\begin{example}\label{ex:seplog}
The model $\mathcal{M}_{\text{\tiny{SL}}} \cap \mathcal{M}_{\text{\tiny{IV}}}$ is identified for the separable logistic missing data mechanism
\begin{eqnarray}\label{mis:lsepr}
\mathcal{M}_{\text{\tiny{SL}}}=\{P(R=1|Z,Y) = \text{expit}[q(Z) + h(Y)]\},
\end{eqnarray}
where  $q(\cdot)$ and $h(\cdot)$ are unknown functions differentiable with respect to $Z$ and $Y$ respectively. 
\end{example}

\setcounter{section}{4} 
\setcounter{equation}{0} 
\noindent {\bf 4. Estimation and Inference}

In this section, we consider estimation and inference under a variety of semiparametric IV models shown to satisfy Result (\ref{thm:para}). We denote the collection of such identifiable models as $\mathcal{M}^{\ast}_{\text{\tiny{IV}}}$. Although in principle the identification results given in the previous section allow for nonparametric inference, in practice estimation often involves specifying parametric models, at least for parts of the full data law. This will generally be the case when a large number of covariates $X$ or $Z$ are present and therefore the curse of dimensionality precludes the use of nonparametric regression to model conditional densities or their mean functions required for IV inferences \citep{rit}. As a measure of departure from MAR, we introduce the selection bias function
\begin{align}\label{bias}
\eta(x,y,z)= \log \left\{ \frac{P(R=1| x,y,z)}{P(R=0| x,y,z)}  /    \frac{P(R=1| x,Y=0,z)}{P(R=0| x,Y=0,z)} \right\}.
\end{align}
$\eta$ quantifies the degree of association between $Y$ and $R$ given $(X,Z)$ on the log odds ratio scale. Under MAR, $P(R=1| x,y,z)=P(R=1| x,z)$ and $\eta=0$. The conditional density $P(r,y|x,z)$ can be represented in terms of the selection bias function $\eta$ and baseline densities as follows:
\begin{align} \label{joint1}
P(r,y|x,z)= C(x,z)^{-1}\exp[(r-1)\eta(x,y,z)] \times\\ \notag 
f(y|R=1,x,z)P(r|Y=0,x,z),
\end{align}
where $C(x,z)<+\infty$ for all $(x,z)$ is a normalizing constant \citep{RePEc:bla:biomet:v:63:y:2007:i:2:p:413-421,10.2307/27798905}. Therefore,
\begin{align} \label{joint}
P(r,y,z|x)=& C(x,z)^{-1}\exp[(r-1)\eta(x,y,z)] \times \\ \notag
&f(y|R=1,x,z)P(r|Y=0,x,z)q(z|x),
\end{align}
where $q(z|x)$ models the density of the IV conditional on the covariates. As we show below, the selection bias function $\eta$ in (\ref{joint}) will need to be correctly specified for any of the three proposed estimators to be consistent. This is significant in that for a given observed data law and  selection bias function $\eta$, one can identify a unique full data law that marginalizes to the observed data law \citep{jamie5}. Absent of restrictions such as Assumption (IV.1), the selection bias function is not identifiable from the observed data law since different values of $\eta$ can lead to the same observed data likelihood. In order to address this identification problem, sensitivity analysis has been previously proposed whereby one conducts inferences assuming $\eta$ is completely known and repeats the analysis upon varying the assumed value of $\eta$ \citep{robins,rotni1,rotni2,rotni3,vans}. A different approach is possible with an IV since $\eta$ is in principle identified under Result \ref{thm:para} and therefore needs not be assumed known. As previously mentioned, it is impossible to disentangle the full data law from the selection process without evaluating $\eta$. Therefore, we will proceed by assuming that although a priori unknown, one can correctly specify a model $\eta(\zeta)$ for the selection bias function which can be estimated from the observed data. To fix ideas, throughout we suppose that one aims to make inferences about the population mean $\phi=E(Y)$, although the proposed methods are easy to extend to other full data functionals.

IPW estimation requires a correctly specified model for the extended propensity score $\pi(x,y,z)$, which under logit link function is
\begin{align} \label{ipwparam}
\pi(x,y,z)= 1/\{1+\exp[-\eta(x,y,z)-\lambda(x,z)]\},
\end{align}
where $\eta(x,y,z)$ is the selection bias function given in (\ref{bias}) and $\lambda(x,z)=\log\{ P(R=1|Y=0,x,z) / P(R=0|Y=0,x,z)\}$ is a person's baseline conditional odds of observing complete data. Although in principle, one could use any well-defined link function for the propensity score, we simplify the presentation by only considering the logit case. We consider IPW estimation in the model $\mathcal{M}_{\text{\tiny{IPW}}} \subset \mathcal{M}^{\ast}_{\text{\tiny{IV}}}$, where
\begin{align*} 
\mathcal{M}_{\text{\tiny{IPW}}}= \biggl\{&P(r,y,z|x) : \eta(x,y,z; \zeta), P(r|Y=0,x,z; \omega) ,q(z|x; \xi); \\
&\text{unrestricted }P(y|R=1,x,z) \biggr\},
\end{align*} 
and the parametric models indexed by parameters $\zeta$, $\omega$ and $\xi$ respectively are assumed to be correctly specified, while the baseline outcome model $f(y|R=1,x,z)$ in (\ref{joint}) is  unrestricted.

Outcome regression-based estimation under MAR requires a model for $f(y|R=1,x,z)=f(y|x,z)$, which can be estimated based on complete-cases. However, under MNAR $f(y|R=1,X,Z)\neq f(y|R=0,X,Z)$ and estimation of $f(y|R=0,x,z)$ is not readily available since outcome is not observed for this subpopulation. However, note that by (\ref{joint1})
\begin{align} \label{orparam}
f(y|r,x,z)&=\frac{P(y,r|x,z)}{\int P(y,r|x,z) \mathrm{d}\mu(y)} =\frac{\exp[-(1-r)\eta(x,y,z)]f(y|R=1,x,z)}{E\{\exp[-(1-r)\eta(x,Y,z)]|R=1,x,z\}},
\end{align}
and therefore the density $f(y|R=0,x,z)$ can be expressed in terms of the selection bias function $\eta$ and baseline outcome model $f(y|R=1,x,z)$ for complete-cases. We consider OR estimation in the model $\mathcal{M}_{\text{\tiny{OR}}} \subset \mathcal{M}^{\ast}_{\text{\tiny{IV}}}$ where
\begin{align*} 
\mathcal{M}_{\text{\tiny{OR}}}= \biggl\{&P(r,y,z|x) : \eta(x,y,z; \zeta), P(y|R=1,x,z; \theta),q(z|x; \xi); \\
&\text{unrestricted }\lambda(x,z) \biggr\},
\end{align*} 
which allows the baseline missing data model $P(r|Y=0,x,z)$ to remain unrestricted while the models indexed by parameters $\zeta$, $\theta$ and $\xi$ are assumed to be correctly specified.

We also propose a doubly robust estimator which is consistent in the union model $\mathcal{M}_{\text{\tiny{IPW}}} \cup \mathcal{M}_{\text{\tiny{OR}}}$, that is provided the models $\eta(x,y,z; \zeta)$ and $q(z|x;\xi)$ are correctly specified, and either $P(r|Y=0,x,z; \omega)$ or $P(y|R=1,x,z; \theta)$, but not necessarily both, are correctly specified, thus giving the analyst two chances, instead of one, to obtain valid inferences.

Throughout the next section, we let $\hat{\theta}_{\text{\tiny{MLE}}}$ denote the complete-case maximum likelihood estimator which maximizes the conditional log-likelihood $\sum_{i: R_i=1}\log P(y_i|x_i,z_i;\theta)$, and let $\hat{\xi}_{\text{\tiny{MLE}}}$ denote the maximum likelihood estimator which maximizes the log-likelihood $\sum_{i=1}^{n} \log q(z_i|x_i; \xi)$. Let $\mathbb{P}_n$ denote the empirical measure $\mathbb{P}_n f(O) = n^{-1} \sum_{i=1}^n f(O_i)$. 

\noindent {\bf 4.1 Inverse probability weighted estimation under $\mathcal{M}_{\text{\tiny{IPW}}}$}

IPW is a well-known approach to acount for missing data under MAR. In this section we describe an analogous approach under MNAR. Standard approaches for estimating the propensity score under MAR such as  maximum likelihood of a logistic regression model of the propensity score cannot be used here since the extended propensity score $\pi(x,y,z)$ depends on $Y$ which is only observed when $R=1$. Therefore, we adopt an alternative method of moments approach which resolves this difficulty. Under the model $\mathcal{M}_{\text{\tiny{IPW}}}$, $(\hat{\zeta},\hat{\omega})$ solves
 \begin{align} \label{IPW}
 \mathbb{P}_n \left\{ \boldsymbol{U}^{\text{\tiny{IPW}}}\left(\hat{\xi}_{\text{\tiny{MLE}}},\hat{\zeta},\hat{\omega} \right) \right\}= \boldsymbol{0}
 \end{align}
 where $\boldsymbol{U}^{\text{\tiny{IPW}}}(\cdot)$ consists of the estimating functions
\begin{align} 
&\left[ \frac{R}{\pi\left(\hat{\zeta},\hat{\omega}\right)}-1 \right] \boldsymbol{h}_1(X,Z) \label{ipw2} \\
&\frac{R}{\pi\left(\hat{\zeta},\hat{\omega}\right)}\boldsymbol{g}(X,Y)\left\{\boldsymbol{h}_2(Z,X)-E\left[\boldsymbol{h}_2(Z,X)\middle\vert X;\hat{\xi}_{\text{\tiny{MLE}}} \right]\right\}. \label{ipw3} 
\end{align}
Functions (\ref{ipw2}) and (\ref{ipw3}) estimate unknown parameters in $P(r|Y=0,x,z; \omega)$ and $\eta(x,y,z;\zeta)$ respectively, { where $\boldsymbol{h}_1$ is an user-specified function of $(x,z)$ with same dimension as $\omega$, while $\boldsymbol{g}$ and $\boldsymbol{h}_2$ are user-specified functions of $(x,y)$ and $(x,z)$ respectively with same dimension as $\zeta$.} Specific choices of $(\boldsymbol{h}_1,\boldsymbol{h}_2,\boldsymbol{g})$ can generally affect efficiency but not consistency. To illustrate IPW estimation, suppose that $Z$ is binary and consider the following logistic model for the extended propensity score
$$
\text{logit } \pi(X,Y,Z)=\omega_0+\omega_1 X+\omega_2 XZ + \zeta Y, \quad \eta=(\omega_0,\omega_1,\omega_2,\zeta).
$$
Thus, $\eta(x,y,z;\zeta)=\zeta y$ and $\text{logit } P(R=1|Y=0,x,z; \omega)=\omega_0+\omega_1 x+\omega_2 xz$. Suppose further that $q(Z=1|x;\xi)=B(x;\xi)=\left\{ 1+ \exp \left[ -(1,x)^T\xi \right]\right\}^{-1}$. We obtain $\hat{\eta}=(\hat{\zeta},\hat{\omega})$ by solving
\begin{align*}
\mathbb{P}_n &\left\{ \left[ \frac{R}{\pi\left(\hat{\zeta},\hat{\omega}\right)}-1 \right] (1,X,XZ)^T\right\}=0 \\  
\mathbb{P}_n  &\left\{ \frac{R}{\pi\left({\hat{\zeta},\hat{\omega}}\right)}Y\left\{Z-B(X;{\hat{\xi}_{\text{\tiny{MLE}}}})\right\}\right\}=0.\\
\end{align*}

\begin{proposition}
Consider a model $\mathcal{M}_{\text{\tiny{IPW}}}\subset \mathcal{M}^{\ast}_{\text{\tiny{IV}}}$ which satisfies Result (\ref{thm:para}). Then the IPW estimator
\begin{align}
\hat{\phi}^{\text{\tiny{IPW}}} = \mathbb{P}_n  \left\{\frac{R Y}{\pi(\hat{\eta})} \right\} 
\end{align}
is consistent and asymptotically normal as $n\rightarrow\infty$, that is
$$
\sqrt{n}\left( \hat{\phi}^{\text{\tiny{IPW}}} - \phi_0 \right)\xrightarrow[]{d} N \left( 0 , V_{IPW} \right) 
$$
in model $\mathcal{M}_{\text{\tiny{IPW}}}$ under suitable regularity conditions, where $V_{IPW}$ is given in the Supplemental Appendix.
\end{proposition}

\noindent {\bf 4.2 Outcome regression estimation under $\mathcal{M}_{\text{\tiny{OR}}}$}

Next, consider inferences under a parametric model for the outcome, i.e. under model $\mathcal{M}_{\text{\tiny{OR}}}$. 
Using the parametrization given in (\ref{orparam}), consider the parametric model
$$
P(y|R=0,x,z; \zeta, \hat{\theta}_{\text{\tiny{MLE}}})=\frac{\exp\left[-\eta(x,y,z;\zeta)\right]f\left(y|R=1,x,z; \hat{\theta}_{\text{\tiny{MLE}}}\right)}{E\left\{\exp[-\eta(x,Y,z; \zeta)]|R=1,x,z;\hat{\theta}_{\text{\tiny{MLE}}}\right\}},
$$
and the estimator $\tilde{\zeta}$ solving
\begin{align} \label{oreqn}
&\mathbb{P}_n \left\{\boldsymbol{U}^{\text{\tiny{OR}}}\left(\tilde{\zeta},\hat{\xi}_{\text{\tiny{MLE}}},\hat{\theta}_{\text{\tiny{MLE}}},\boldsymbol{q}_1, \boldsymbol{q}_2\right)\right\} \nonumber \\
&=\mathbb{P}_n \left\{\boldsymbol{q}_1(X,Z)-E\left[\boldsymbol{q}_1(X,Z)\middle\vert X;\hat{\xi}_{\text{\tiny{MLE}}}\right] \right\} \nonumber \times \\ \nonumber
&\phantom{=}\left\{ (1-R)  E\left(\boldsymbol{q}_2(X,Y)\middle\vert R=0,X,Z;\tilde{\zeta},\hat{\theta}_{\text{\tiny{MLE}}}\right)           +R\boldsymbol{q}_2(X,Y) \right\} \nonumber \\
&=\boldsymbol{0},
\end{align}
where $\boldsymbol{q}_1, \boldsymbol{q}_2$ are vectors of the same dimensions as $\zeta$.

\begin{proposition}
Consider a model $\mathcal{M}_{\text{\tiny{OR}}}\subset \mathcal{M}^{\ast}_{\text{\tiny{IV}}}$ which satisfies Result (\ref{thm:para}). Then the outcome regression estimator
\begin{align}
\hat{\phi}^{\text{OR}} = \mathbb{P}_n  \left\{R Y + (1-R)E\left(Y\middle\vert R=0,X,Z;\tilde{\zeta},\hat{\theta}_{\text{\tiny{MLE}}}\right)\right\} 
\end{align}
is consistent and asymptotically normal as $n\rightarrow\infty$, that is
$$
\sqrt{n}\left( \hat{\phi}^{\text{OR}} - \phi_0 \right)\xrightarrow[]{d} N \left( 0 , V_{OR} \right)
$$
in model $\mathcal{M}_{\text{\tiny{OR}}}$ under suitable regularity conditions.
\end{proposition}

\noindent {\bf 4.3 Doubly robust estimation under $\mathcal{M}_{\text{\tiny{DR}}}$}

Estimation approaches described thus far depend on correct specification of extended propensity score for IPW and outcome model for OR. Here we describe a doubly robust estimator that remains consistent if the conditional density $q(z|x; \xi)$ is correctly specified, and either $P(y|R,X,Z; \theta)$ or $P(r|Y,X,Z; \omega)$ is correctly specified, but not necessarily both. We denote such union model $\mathcal{M}_{\text{\tiny{DR}}}=\mathcal{M}_{\text{\tiny{IPW}}}\cup\mathcal{M}_{\text{\tiny{OR}}}$. Our construction requires first obtaining the DR estimator $\hat{\zeta}_{\text{\tiny{DR}}}$ of the parameter indexing selection bias function $\eta(\zeta)$ that remains consistent in $\mathcal{M}_{\text{\tiny{DR}}}$. In this vein, let
\begin{align}
 &\boldsymbol{G}^{\text{\tiny{DR}}}\left(R,X,Y,Z; {\zeta},\omega,\hat{\theta}_{\text{\tiny{MLE}}},\boldsymbol{u}\right) \nonumber\\
 &= \frac{R}{\pi({\zeta},{\omega})} \boldsymbol{u}(X,Y)- \frac{R-\pi({\zeta},{\omega})}{\pi({\zeta},{\omega})}  E\left(\boldsymbol{u}(X,Y)\middle\vert R=0,X,Z;{\zeta},\hat{\theta}_{\text{\tiny{MLE}}}\right)  \nonumber\\
 &=\frac{R}{\pi({\zeta},{\omega})} \left\{\boldsymbol{u}(X,Y)-E\left(\boldsymbol{u}(X,Y)\middle\vert R=0,X,Z;{\zeta},\hat{\theta}_{\text{\tiny{MLE}}}\right) \right\} \nonumber\\
 &+ E\left(\boldsymbol{u}(X,Y)\middle\vert R=0,X,Z;{\zeta},\hat{\theta}_{\text{\tiny{MLE}}}\right),
\end{align}
where $\boldsymbol{u}(X,Y)$ is of the same dimensions as $\zeta$. We obtain $(\hat{\zeta}_{\text{\tiny{DR}}},\hat{\omega})$ as the solution to the estimating equation (\ref{ipw2}) combined with
\begin{align} \label{tref}
&\mathbb{P}_n  \left\{\boldsymbol{U}^{\text{\tiny{DR}}}\left( \hat{\zeta}_{\text{\tiny{DR}}},\hat{\omega},\hat{\theta}_{\text{\tiny{MLE}}},\hat{\xi}_{\text{\tiny{MLE}}},\boldsymbol{u},\boldsymbol{v} \right)\right\}\nonumber \\
&=\mathbb{P}_n \bigg\{  \left[\boldsymbol{v}(X,Z)-E\left(\boldsymbol{v}(X,Z)\middle\vert X;\hat{\xi}_{\text{\tiny{MLE}}}\right)\right] \times \nonumber \\ 
&\phantom{=}\left[ {\boldsymbol{G}}^{\text{\tiny{DR}}}\left(R,X,Y,Z;\hat{\zeta}_{\text{\tiny{DR}}},\hat{\omega},\hat{\theta}_{\text{\tiny{MLE}}},\boldsymbol{u}\right)  \right] \bigg\} \nonumber\\
&=\boldsymbol{0}.
\end{align}

\begin{proposition}
The laws in $\mathcal{M}_{\text{\tiny{DR}}}\subset \mathcal{M}^{\ast}_{\text{\tiny{IV}}}$ satisfies Result (\ref{thm:para}). Then the doubly robust estimator
\begin{align} \label{tr}
\hat{\phi}^{\text{\tiny{DR}}}=\mathbb{P}_n \left\{\boldsymbol{G}^{\text{\tiny{DR}}}\left(R,X,Y,Z,\hat{\zeta}_{\text{\tiny{DR}}},\hat{\omega},\hat{\theta}_{\text{\tiny{MLE}}},\boldsymbol{u}^{\dagger}\right)\right\}
\end{align}
where $\boldsymbol{u}^{\dagger}(X,Y)=Y$ is consistent and asymptotically normal as $n\rightarrow \infty$, that is
$$
\sqrt{n}\left( \hat{\phi}^{\text{\tiny{DR}}} - \phi_0 \right)\xrightarrow[]{d} N \left( 0 , V_{DR} \right)
$$
in the model $\mathcal{M}_{\text{\tiny{DR}}}$ under suitable regularity conditions.
\end{proposition}
The notion of doubly robust estimation was first introduced in the context of semi-parametric non-response models under MAR \citep{rotni3}, and the approach was further studied by others \citep{lip,robins,luc,neu} with theoretical underpinnings given by \citet{robi} and \citet{vander}. A doubly robust version of estimating equation (\ref{tr}) of mean outcome under MNAR was previously described by \citet{vans} who, as described earlier, assume that the selection bias function $\eta$ is known a priori within the context of a sensitivity analysis. An important contribution of the current paper is to derive a large class of DR estimators of the selection bias using an IV. To the best of our knowledge, this is the first time that a DR estimator for the mean outcome has been constructed in the context of an IV for data subject to MNAR.

\setcounter{section}{5} 
\setcounter{equation}{0} 
\noindent {\bf 5. Simulation study}

In order to investigate the finite-sample performance of  proposed estimators, we carried out a simulation study involving i.i.d. data $\left(Y,Z,X\right)$, where $X=(X_1,X_2)$. For each sample size $n=2000, 5000$, we simulated 1000 data sets as followed,
\begin{align*}
&X_1 \sim \text{Bernoulli} (p=0.4), \quad X_2 \sim \text{Bernoulli} (p=0.6) \\
& Z \sim \text{Bernoulli} \left\{p=\left[ 1+ \exp \left( -0.4-0.9X_1+0.7X_2+0.8X_1X_2\right)\right]^{-1}\right\}\\
& Y \sim \text{Bernoulli} \left\{p=\left[ 1+ \exp \left( -1.0+1.2X_1-1.5X_2\right)\right]^{-1}\right\}\\
& R \sim \text{Bernoulli} \left\{p=\left[ 1+ \exp \left(1.5 -2.5Z-0.8X_1+1.2X_2-1.8Y\right)\right]^{-1}\right\},
\end{align*}
such that $Y$ is only observed if $R=1$. Under the above data generating mechanism, $Z$ satisfies {\bf(IV.1)} and {\bf(IV.2)}, with the true value of $\phi_0=E(Y)=0.769$. The selection bias model is $\alpha(x,y,z)=\zeta y$ with true value $\zeta_0=1.8$. The model is identified since the missing data mechanism follows the separable logistic regression model described in Example \ref{ex:seplog} of Section 3. For IPW estimation, we specified the correct extended propensity score and model for $P(Z=1|X_1,X_2; \xi)$, with $h_1=(Z,X_1,ZX_1)^T$, $g=Y$ and $h_2=Z$. For OR estimation, we let $(q_1,q_2) = (Z,Y)$ in (\ref{oreqn}) and specified a saturated logistic regression for $Y$ with all 2-way and 3-way interactions included. DR estimation was carried out as described in the previous section. {While} \citet{chang2008using} {only considered a survey design setting, we note that here the IPW approach is analogous to a form of calibration weighted estimation which matches the weighted sample estimates of benchmark variables $L_{CW}=\left\{1,Z, X_1,X_2, Y\left[Z-P(Z=1|X_1,X_2)\right] \right\}$ to their estimated population totals, where the last variable in $L_{CW}$ has known population total value of zero by {\bf(IV.1)}.}

To study the performance of the proposed estimators in situations where some models may be mis-specified, we also evaluated the estimators where either the extended propensity score model or the complete-case outcome model was mis-specified by replacing them with models
\begin{align*}
P(R=1|X,Y,Z) = \text{expit} (\omega_0+\omega_1 X_1 + \omega_2 Z + \omega_3 X_1Z +\zeta Y)
\end{align*}
and
\begin{align*}
P(Y=1|R=1, X,Z)=\text{expit} (\theta_0 + \theta X_1)
\end{align*}
respectively. 

In each simulated sample, we evaluated the standard error of the estimator using the sandwich estimator. Wald 95\% confidence interval coverage rates were evaluated across 1000 simulations. Estimating equations were solved using the R package BB \citep{bb}. Figures \ref{fig:zeta} and \ref{fig:phi} present results for estimation of the selection bias parameter $\zeta_0$ and the outcome mean $\phi_0$ respectively, while Table \ref{tab:simulation} shows the empirical coverage rates.

\begin{figure}[H]
   \centering
    \caption{Boxplots of inverse probability weighted (IPW), outcome regression (OR) and doubly-robust (DR) estimators of the selection bias parameter, for which the true value $\zeta_0=1.8$ is marked by the horizontal lines.}
    
       \includegraphics[page=1,width=\textwidth]{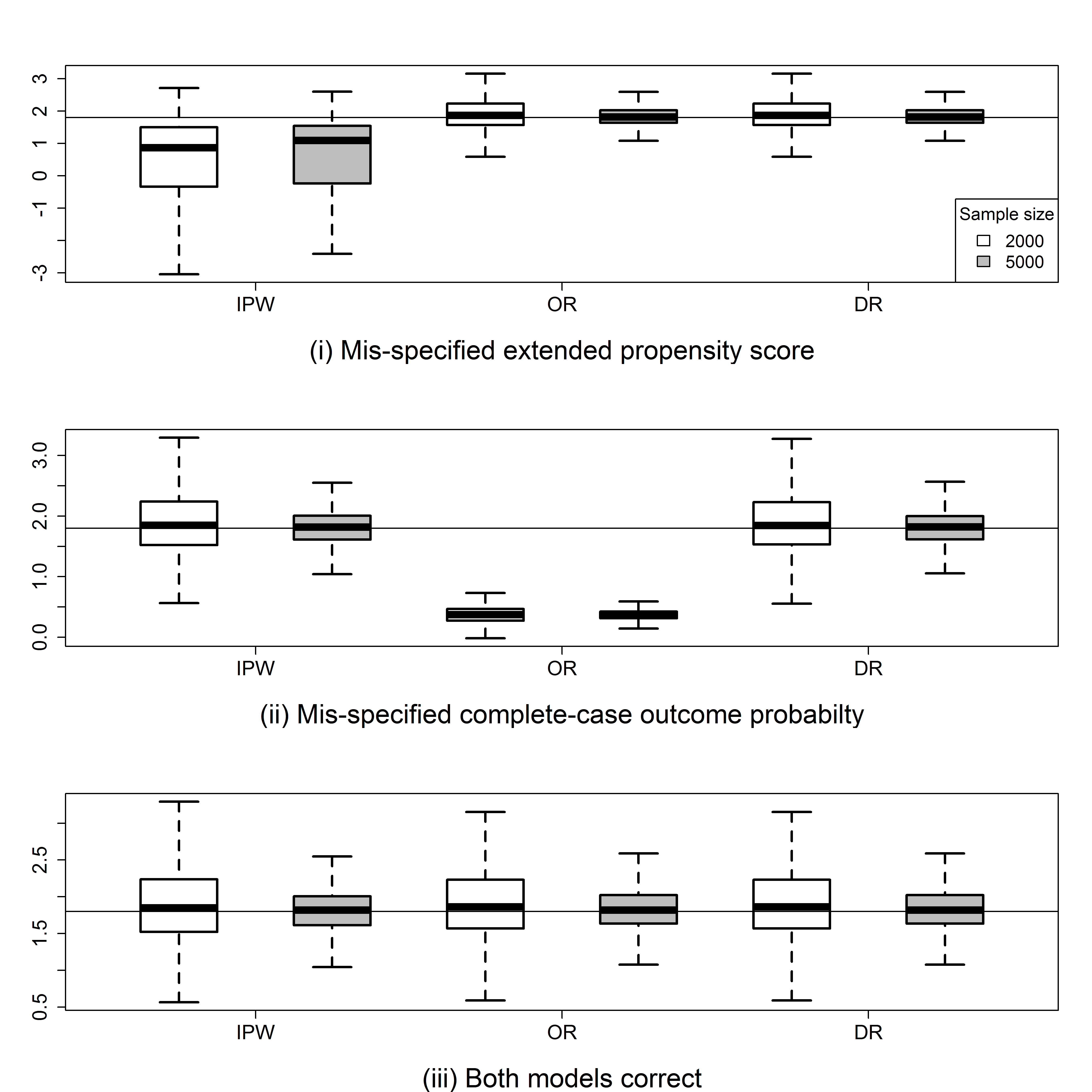} 
        \label{fig:zeta}
\end{figure}

\begin{figure}[H]

   \centering
    \caption{Boxplots of inverse probability weighted (IPW), outcome regression (OR) and doubly-robust (DR) estimators of the outcome mean, for which the true value $\phi_0=0.769$ is marked by the horizontal lines.}
   \begin{tabular}{@{}c@{\hspace{.1cm}}c@{}}
       \includegraphics[page=1,width=\textwidth]{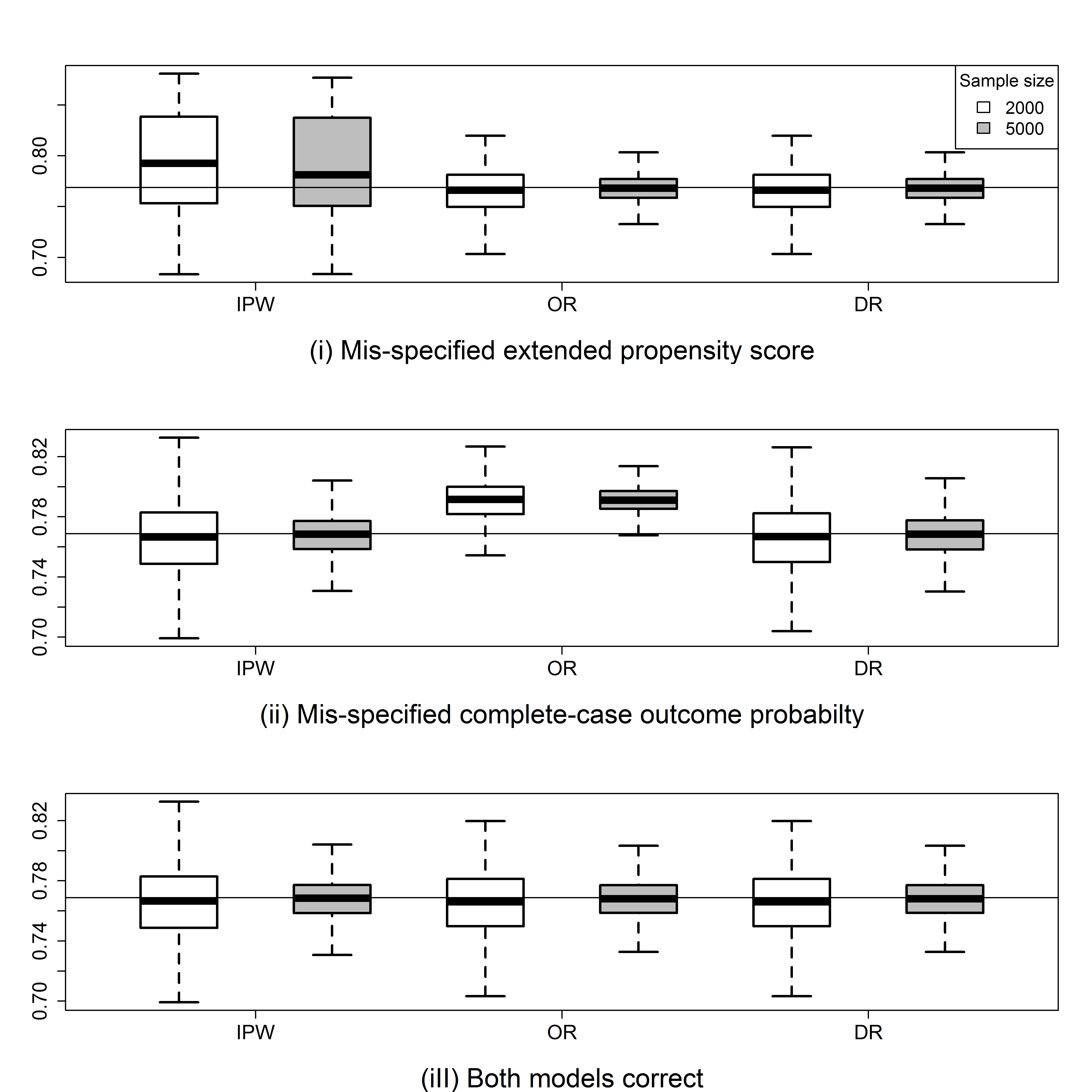} 
   \end{tabular}
 \label{fig:phi}
\end{figure}

\begin{table}[t!]
\caption{\small{Empirical coverage rates based on 95\% Wald confidence intervals under three scenarios: (i) mis-specified extended propensity score, (ii) mis-specified complete-case outcome probability and (iii) both models are correct. In each scenario, the first row presents results for $n=2000$ and the second row for $n=5000$.}}
\label{tab:simulation}\par
\vskip .2cm
\centering
\renewcommand{\arraystretch}{0.7}
\begin{tabular}{c c c c c c c c}
\hline\noalign{\smallskip}
& \multicolumn{3}{c}{$\zeta$} & \multicolumn{3}{c}{$\phi$} \\
\cmidrule(lr){2-4}\cmidrule(lr){5-7}
&IPW & OR & DR & IPW & OR & DR \\
\hline\noalign{\smallskip}
\multirow{2}{*}{(i)} & 86.4 & 95.4 & 95.4 & 81.3 & 95.2 & 95.2 \\ 
& 57.8 & 95.1 & 95.1 & 50.1 & 94.9 & 94.9 \\
\multirow{2}{*}{(ii)} & 95.0 & 0.0 & 94.4 & 95.1 & 65.6 & 95.2 \\ 
& 94.7 & 0.0 & 94.5 & 95.0 & 29.9 & 94.5 \\
\multirow{2}{*}{(iii)} & 95.0 & 95.4 & 95.4 & 95.1 & 95.2 & 95.2 \\ 
& 94.7 & 95.1 & 95.1 & 95.0 & 94.9 & 94.9 \\
\hline
\end{tabular}

\end{table}

Under correct model specification, all estimators have negligible bias for $\phi_0$ and $\zeta_0$ that diminishes with increasing sample size, with empirical coverage near the nominal 95\% level. In agreement with our theoretical results, the IPW and OR estimators are biased with poor empirical coverages when the extended propensity score or the complete-case outcome model is misspecified, respectively. The DR estimator performs well in terms of bias and coverage when either model is misspecified but the other is correct. 

\setcounter{section}{6} 
\setcounter{equation}{0} 
\noindent {\bf 6. Applications}

To illustrate the proposed IV approach, we obtained data from a household survey in Mochudi, Botswana to estimate HIV seroprevalence among adults adjusting for selective missingness of HIV test results. The data consist of 4997 adults between the ages of 16 and 64 who were contacted for the survey, out of whom 4045 (81\%) had complete information on HIV testing. Of those who did not have HIV test results $(R=0)$, 111 (2\%) agreed to participate in the HIV test but their final HIV outcomes are unknown, and 841 (17\%) refused to participate in the HIV testing component. It is likely that refusal to participate in the survey when contact is established presents a possible source of selection bias. 

Fully available individual characteristics from the survey include participant gender $(X)$. Candidate IVs include interviewer gender $(Z_1)$, age $(Z_2)$ and years of experience $(Z_3)$. These interviewer characteristics are likely to influence the response rates of individuals who were contacted for the survey, but are unlikely to directly influence an individual's HIV status, given that interviewer deployment was determined at random prior to the survey. We implemented the proposed IPW, OR and DR estimators by making use of interviewer gender, age and years of experience as IVs. For IPW estimation, the missingness propensity score is specified as a main effects only logistic regression, with the selection bias function specified as $\alpha(x,y,z)=\zeta y$ where $Y$ is HIV status. The posited missing data mechanism belongs to the separable logistic class, therefore the average HIV prevalence can be identified by Example \ref{ex:seplog}. For OR estimation, we specified the regression model
\begin{align}\label{app:out}
\text{logit } P(Y=1|R=1,X,\boldsymbol{Z}) = \theta_0 + \theta_1 X + \theta_2 Z_1 + \theta_3 Z_2 + \theta_4 Z_3. 
\end{align}
Finally, the doubly robust estimator is implemented by incorporating both models. Because more than one IV was available, estimating equations $\boldsymbol{U}^{\text{\tiny{IPW}}}$, $\boldsymbol{U}^{\text{\tiny{OR}}}$ and $\boldsymbol{U}^{\text{\tiny{DR}}}$ were solved using the generalized method of moments (GMM) package in R \citep{gmm}.  Standard errors were obtained using the proposed sandwich estimator. For comparison, we also carried out standard complete-case analysis and standard IPW estimation assuming MAR conditional on $(x,z)$ using a main effects only logistic regression to model the propensity score. Results are presented in table \ref{tab:application}.

\begin{table}[h!]
\caption{\small{Estimation for HIV seroprevalence ($\phi$) and magnitude of selection bias ($\zeta$) in Mochudi, Botswana with 95\% Wald confidence intervals.}}
\label{tab:application}\par
\vskip .2cm
\centering

\begin{tabular}{c c  c c}
\hline
{Estimator} & {$\hat{\phi}$}   &  {$\hat{\zeta}$} & {$\hat{\zeta}$ p-val}\\
\hline
CC & 0.214  (0.202, 0.227)&   - &  -  \\ 
MAR IPW& 0.213 (0.201, 0.226) & - & -  \\
IV IPW& 0.260 (0.175, 0.341)&  -1.601 (-2.992, -0.210) &  0.02 \\ 
IV OR& 0.241 (0.175, 0.307)&  -0.757 (-1.889, \phantom{-}0.376)  &  0.19\\ 
IV DR& 0.258 (0.174, 0.342)& -1.121 (-2.433, \phantom{-}0.191)  &  0.09\\
 \hline
\end{tabular}
\end{table}
IV estimates of HIV seroprevalence are $12.6-21.5\%$ higher than the crude estimate of 0.214 (95\% CI: 0.202-0.227) based on complete-cases only. Standard IPW (i.e. assuming MAR) produced similar estimates as complete-case analysis. Negative point estimates of the selection bias parameter $\zeta$ suggest that HIV-infected persons are less likely to participate in the HIV testing component of the survey, although this difference is statistically significant at $0.05$ $\alpha$-level only for IPW. The larger confidence intervals of the three IV estimators of $\phi_0$ compared to those of the CC and MAR estimators are a more accurate reflection of the amount of uncertainty involving inferences about $\phi_0$, since the CC and MAR estimators do not take into account the uncertainty about the underlying MNAR mechanism by assuming MCAR and MAR respectively, i.e. setting selection bias parameter $\zeta=0$. $\hat{\phi}^{\text{\tiny{IV IPW}}}$ and $\hat{\phi}^{\text{\tiny{IV DR}}}$ are close to each other. This comparison is useful as an informal goodness of fit test in that their similarity suggests that the missingness propensity score may be specified nearly correctly \citep{robi}. In addition, by incorporating all possible pairwise interaction terms in the outcome logistic regression model and therefore allowing it to be more flexible, the OR point estimate $\hat{\phi}^{\text{\tiny{IV OR}}}$ increases to 0.246 (95\% CI: 0.179-0.314), thus even closer to $\hat{\phi}^{\text{\tiny{IV IPW}}}$ and $\hat{\phi}^{\text{\tiny{IV DR}}}$.

\noindent {\bf 7. Conclusion}

In this paper, we have considered a pernicious form of selection bias which
can arise from outcome missing not at random. We have argued that under
fairly reasonable assumptions this problem can be made more tractable
with the aid of an IV, and proposed a general framework for establishing
identifiability of parametric, semiparametric and nonparametric models. We
have proposed IPW and OR estimators which are consistent and asymptotically
normal if the selection bias and the IV models are correctly specified,
when either the extended propensity score or the outcome regression model
is correctly specified respectively. We also constructed a DR estimator that
remains consistent if either of the two models is correct, which gives the
analyst two chances, instead of only one, to get correct inferences about
the magnitude of selection bias and the mean outcome in the underlying
population of interest.

The large sample variance of doubly robust estimators $\widehat{\zeta }_{DR}$
and $\widehat{\phi }_{DR}$ at the intersection submodel $\mathcal{M}%
_{IPW}\cap \mathcal{M}_{OR}$ where all models are correct, is completely
determined by the choice of $\boldsymbol{u}$ and $\boldsymbol{v}$ in
equation (\ref{tref}). We have characterized the set of all influence functions of regular and asymptotically linear estimators as well as the semiparametric efficient score of $\left( \zeta ,\phi \right) $ in model\ $\mathcal{M}_{np}$ that assumes that $Z$ is a valid IV, the selection bias function $\eta \left( X,Y,Z;\zeta \right)$ is correctly specified, and
the joint likelihood of $\left( Y,X,Z,R\right)$ is otherwise unrestricted. The efficient score is not generally available in closed-form, except in special cases, such as when $Z$ and $Y$ are both polytomous. The results on local efficiency results can be found in the Supplemental Appendix.

\markboth{\hfill{\footnotesize\rm FIRSTNAME1 LASTNAME1 AND FIRSTNAME2 LASTNAME2} \hfill}
{\hfill {\footnotesize\rm Identification and DR Estimation for Outcome MNAR} \hfill}

\bibhang=1.7pc
\bibsep=2pt
\fontsize{9}{14pt plus.8pt minus .6pt}\selectfont
\renewcommand\bibname{\large \bf References}
\renewcommand{\thepage}{}

\bibliographystyle{apa}
\bibliography{mnarbib}
\vskip .65cm
\noindent
\vskip 2pt
\noindent
\vskip 2pt

\noindent
\vskip 2pt
\noindent
\newpage

\setcounter{page}{1}
\fontsize{12}{14pt plus.8pt minus .6pt}\selectfont
\pagestyle{plain}
\pagenumbering{arabic} 
\renewcommand{\thepage}{S\arabic{page}}
\noindent{\bf\Large Supplemental Appendix}

\noindent The appendix includes proofs for Result, Examples and Propositions (pp. S1-12), results on local efficiency (pp. 13-23) as well as R code for simulation study (pp. 24-46).
\section*{Proof of Result 1}

The proof is based on contradiction. By the exclusion restriction assumption {\bf(IV.1)} the decomposition of the joint distribution for $(Z,Y,R)$ is 
$$
P_{\theta_i,\eta_i,\xi_i}(z,y,r)=P_{\theta_i}(r|z,y)P_{\eta_i}(z)P_{\xi_i}(y), \quad i=1,2,...,n
$$
Suppose we have two sets of candidates satisfying the same observed quantities:
\begin{align*}
P_{\theta_1}(z,y,R=1) &= P_{\theta_2}(z,y,R=1) \\
P_{\eta_1}(z) &= P_{\eta_2}(z) 
\end{align*}

Substituting the above observed quantities into the joint distribution gives
$$
\frac{P_{\theta_1}(R=1|z,y)}{P_{\theta_2}(R=1|z,y)}=\frac{P_{\xi_2}(y) }{P_{\xi_1}(y) }
$$
This contradicts with the requirement that the ratios are unequal.

\section*{Proofs of Examples 1 and 2}
\setcounter{equation}{0}

\noindent \textbf{Proof of Example 1}

For binary outcome $Y$ and binary instrument $Z$, let $P(R=1|Z,Y; \theta)=\text{expit} \left[\theta_0 + \theta_1 Z + \theta_2 Y + \theta_3 ZY\right]$ and $P(Y=1;\xi)=\exp(\xi)$. We show that for every $(\theta,\xi)$, there exists $(\tilde{\theta},\tilde{\xi}) \neq (\theta,\xi)$ such that 
\begin{align*}
\frac{P(R=1|Z,Y;\theta)}{P(R=1|Z,Y;\tilde{\theta})}=\frac{P(Y;\tilde{\xi}) }{P(Y;\xi) } \tag{A}
\end{align*}
Let $\frac{P(Y=0;\tilde{\xi}) }{P(Y=0;\xi) }=\exp(\rho_0)$ for some $\rho_0 \neq 0$, then $\frac{P(Y;\tilde{\xi}) }{P(Y;\xi) }=\exp(\rho_0+\rho_1 Y)$ where
$$
\rho_1 = \log \left\{ \exp(-\rho_0-\xi)+[\exp(\xi)-1]/\exp(\xi)  \right\}.
$$
Equality (A) then holds by choosing $(\tilde{\theta},\tilde{\xi})$ such that 
\begin{align*}
\tilde{\theta}_0 &={\theta}_0-\rho_0-\log(\alpha_0) \\
\tilde{\theta}_1 &={\theta}_1+\log(\alpha_0)-\log(\alpha_1) \\
\tilde{\theta}_2 &={\theta}_2-\rho_1+\log(\alpha_0)-\log(\alpha_2) \\
\tilde{\theta}_3 &={\theta}_3+\log(\alpha_1)+\log(\alpha_2)-\log(\alpha_0)-\log(\alpha_3) \\
\tilde{\xi} &={\xi}+\rho_0+\rho_1, 
\end{align*}
where $\alpha_0 = 1+\exp(\theta_0)-\exp(\theta_0-\rho_0)$, $\alpha_1 = 1+\exp(\theta_0+\theta_1)-\exp(\theta_0+\theta_1-\rho_0)$, $\alpha_2 = 1+\exp(\theta_0+\theta_2)-\exp(\theta_0+\theta_2-\rho_0-\rho_1)$ and  $\alpha_3 = 1+\exp(\theta_0+\theta_1+\theta_2+\theta_3)-\exp(\theta_0+\theta_1+\theta_2+\theta_3-\rho_0-\rho_1)$. For example, choose $(\rho_0, \rho_1) =(0.3, -0.38)$ and equality (A) holds for $({\theta}_0,{\theta}_1,{\theta}_2,{\theta}_3,{\xi})=(0.3,0.6,0.1,0.7,-0.2)$ and $(\tilde{\theta}_0,\tilde{\theta}_1,\tilde{\theta}_2,\tilde{\theta}_3,\tilde{\xi})=(-0.3,0.41,0.91,1.37,-0.28)$.

Next, we consider the missingness mechanism $P(R=1|Z,Y; \theta)=\text{expit} \left[\theta_0 + \theta_1 Z + \theta_2 Y \right]$, where the interaction effect between $(Z,Y)$ is absent. Under this mechanism, we have $\theta_3=\tilde{\theta}_3 =0$ and therefore $\alpha_1\alpha_2=\alpha_0\alpha_3$ which implies the equality
\begin{align*}
\exp(\rho_0+\rho_1) = \frac{\exp(\theta_2+\rho_0)}{\exp(\theta_2+\rho_0)+[1-\exp(\rho_0)]}. \tag{B}
\end{align*}
Since $\exp(\rho_0+\rho_1 Y)$ is the ratio of the two probability mass distributions for $Y$, $\rho_0$ and $\rho_0+\rho_1$ should be of opposite signs. Based on (B), if $\exp(\rho_0)>1$ then $\exp(\rho_0+\rho_1)>1$ and similarly if $\exp(\rho_0)<1$ then $\exp(\rho_0+\rho_1)<1$, which implies that the only possibility is $\rho_0=\rho_1=0$ and hence $(\tilde{\theta},\tilde{\xi}) = (\theta,\xi)$.

\lhead[\footnotesize\thepage\fancyplain{}\leftmark]{}\rhead[]{\fancyplain{}\rightmark\footnotesize\thepage}

\noindent \textbf{Proof of Example 2}

Consider the case where $Z$ and $Y$ are both continuous random variables. Suppose two sets of candidates in the separable logistic missing data mechanism has the following relationship
$$
\frac{\text{expit}(q_1(z)+h_1(y))}{\text{expit}(q_2(z)+h_2(y))}=g(y)
$$
for some function $g(\cdot)$, i.e. the ratio is a function of $y$ only. Taking derivative with respect to $Z$ on both sides (assuming IV relevance {\bf(IV.2)} holds) gives
$$
\frac{\frac{\partial}{\partial z}\text{expit}(q_1(z)+h_1(y))}{\text{expit}(q_1(z)+h_1(y))}=\frac{\frac{\partial }{\partial z}\text{expit}(q_2(z)+h_2(y))}{\text{expit}(q_2(z)+h_2(y))}
$$
or equivalently
\begin{align*}
\frac{\partial q_1(z) / \partial z}{\partial q_2(z) / \partial z}= \frac{1+\exp(q_1(z)+h_1(y))}{1+\exp(q_2(z)+h_2(y))} \tag{A}
\end{align*}
Taking derivatives with respect to $Y$ on both sides leads to
$$
\frac{\partial q_1(z) / \partial z}{\partial q_2(z) / \partial z} \exp(q_2(z)-q_1(z)) = \frac{\partial h_1(y) / \partial y}{\partial h_2(y) / \partial y} \exp(h_1(y)-h_2(y))
$$
The left hand side of the above equation depends only on $Z$ but the right hand side depends only on $Y$, so it must be that
$$
\frac{\partial q_1(z) / \partial z}{\partial q_2(z) / \partial z} \exp(q_2(z)-q_1(z)) =c_1
$$
for some constant $c_1$. Substituting the above expression into equality (A) leads to
$$
c_1 \left\{ \exp(-q_2(z))+\exp(h_2(y))  \right\} = \exp(-q_1(z))+\exp(h_1(y))
$$
and therefore
$$
c_1 \exp(-q_2(z))+c_2 = \exp((-q_1(z)), \quad c_1 \exp(h_2(y))-c_2 = \exp((h_1(y))
$$
for some constant $c_2$. Substituting the above equalities into the ratio of propensity scores
$$
\frac{\text{expit}(q_1(z)+h_1(y))}{\text{expit}(q_2(z)+h_2(y))}=1+c_2 \exp(-h_1(y)) = g(y)
$$
Note that $g(y)$ is the ratio of two candidate densities of $Y$, and so it must be that $c_2=0$ and the two sets of candidates are equivalent, leading to a contradiction. Therefore the ratio 
$$
\frac{\text{expit}(q_1(z)+h_1(y))}{\text{expit}(q_2(z)+h_2(y))}
$$
is either a constant or depends on $z$, which by Corollary 1 leads to identifiability of this class of missing data models.

Consider the case where $Z$ is a binary random variable, and assume two sets of candidates in the separable logistic missing data mechanism has the following relationship
$$
\frac{\text{expit}(\eta_1 z+h_1(y))}{\text{expit}(\eta_2 z+h_2(y))}=g(y).
$$
The above relationship holds for $z=0,1$, therefore
$$
\frac{\text{expit}(h_1(y))}{\text{expit}(h_2(y))}=\frac{\text{expit}(\eta_1 + h_1(y))}{\text{expit}(\eta_2 + h_2(y))}
$$
and 
$$
g(y) = 1 + \frac{\exp(\eta_2)-\exp(\eta_1)}{\exp(\eta_2)-\exp(\eta_1+\eta_2)}\exp[-h_2(y)].
$$
Since $g(y)$ is the ratio of two densities, we must have $\eta_1=\eta_2$ and $g(y)=1$, leading to a contradiction.
The proof for Y or Z as discrete variables is similar to the above proof for binary $Z$.

\section*{Proofs of Propositions}
\setcounter{equation}{0}

\noindent \textbf{Proof of Proposition 1}

Let $(\eta_0, \omega_0,\xi_0)$ denote the true values of the parameters for parametric models $\eta(x,y,z; \zeta), P(r|Y=0,x,z; \omega)$ and $q(z|x; \xi)$ which are assumed to be correctly specified. Assume the model $q(z|x; \xi)$ is identifiable, its parameter space $\Xi$ is compact and the remaining conditions in Theorem 2.5 of \citet{newey} hold, which are sufficient to establish consistency of maximum likelihood estimators. Then $\hat{\xi}_{\text{\tiny{MLE}}}$ has a probability limit equal to $\xi_0$. Consider estimating function for (4.7) which under the law of iterated expectations equals to 
\begin{align*}
&E\left\{E\left\{ \left[ \frac{R}{\pi(\zeta_0,\omega_0)}-1 \right] \boldsymbol{h}_1(X,Z)\right\}\middle\vert X,Y,Z\right\}\\
=&E\left\{E\left\{ \left[ \frac{\pi(\zeta_0,\omega_0)}{\pi(\zeta_0,\omega_0)}-1 \right] \boldsymbol{h}_1(X,Z)\right\}\right\}=0.
\end{align*}
Under the law of iterated expectations, the estimating function for (4.8) equals
\begin{align*}
& E \left\{ \frac{R}{\pi(\zeta_0,\omega_0)}g(Y,X)\{ h_2(Z,X)-E[h_2(Z,X) \vert X;\xi_0] \} \right\} \\ 
=& E \left\{ g(Y,X)\{ h_2(Z,X)-E[h_2(Z,X) \middle\vert X;\xi_0] \} \right\} \\ 
=& E\left\{ E [ g(Y,X)\middle\vert X ] \{ h_2(Z,X)-E[h_2(Z,X) \middle\vert X;\xi_0] \}  \right\}  \quad \text{by {\bf (IV.1)}}  \\ 
= &E \left\{  E [ g(Y,X)\middle\vert X ] \{ E[h_2(Z,X) \middle\vert X;\xi_0]-E[h_2(Z,X) \middle\vert X;\xi_0] \} \right\}. 
=&0.
\end{align*}
Therefore  $(\eta_0, \omega_0)$ are the probability limits of the solutions to estimating equations (4.7) and (4.8). The IPW estimator is also unbiased,
$$
E  \left\{\frac{R Y}{\pi(\zeta_0,\omega_0)} \right\} =E\{ Y\}=\phi_0,
$$
by taking iterated expectations with respect to $(X,Y,Z)$.
The consistency and asymptotic normality of $\hat{\phi}^{\text{IPW}}$ can be established under standard regularity conditions for GMM estimators \citep{newey} , typically by placing moment restrictions on the vector of estimating functions. In particular, we require that the probability of observing the outcome is bounded away from zero, a necessary assumption for identification of a full data functional \citep{jamie2}.
\begin{align}\label{pos}
\pi (x,y,z) >\sigma > 0 \quad \text{with probability 1} \tag{S1}
\end{align}
for a non-zero positive constant $\sigma>0$. 

Let $\boldsymbol{M}(\delta)$ represent the stacked vector of the following estimating functions: score functions for estimating $\xi$, $\boldsymbol{U}^{\text{\tiny{IPW}}}(\xi,\zeta,\omega)$ and $G(\phi,\zeta,\omega)$\\$=\left\{\frac{R Y}{\pi({\zeta,\omega})} -{\phi}\right\}$, where $\delta=(\zeta,\omega,\xi,\phi)$. Then under standard regularity conditions for M-estimation \citep{newey},  the asymptotic variance $V$ is given by the diagonal entry corresponding to $\phi$ of the following variance-covariance matrix
\begin{align} \label{var}
\left[ E \left\{ \frac{\partial \boldsymbol{M}(\delta)}{\partial \delta^T} \bigg|_{\delta_0} \right\} \right]^{-1}E\left\{\boldsymbol{M}(\delta_0) \boldsymbol{M}(\delta_0) ^T\right\}\left[ E \left\{ \frac{\partial \boldsymbol{M}(\delta)}{\partial \delta^T} \right\}\bigg|_{\delta_0}  \right]^{-1^T}, \tag{S2}
\end{align}
where $\delta_0 = (\zeta_0,\omega_0,\xi_0,\phi_0)$ is the probability limit of $\hat{\delta}=(\hat{\zeta},\hat{\omega},\hat{\xi},\hat{\phi})$. A consistent sandwich estimator for the above asymptotic variance can be constructed by evaluating unknown expectations as sample means at the estimated parameter value $\hat{\delta}$.

\noindent \textbf{Proof of Proposition 2}

Let $(\eta_0, \theta_0,\xi_0)$ denote the true values of the parameters for parametric models $\eta(x,y,z; \zeta), f(y|R=1,x,z; \theta)$ and $q(z|x; \xi)$ which are assumed to be correctly specified. Assume the conditions in Theorem 2.5 of \citet{newey} hold for models $f(y|R=1,x,z; \theta)$ and $q(z|x; \xi)$. Then the probability limits of the MLEs $(\hat{\theta}_{\text{\tiny{MLE}}},\hat{\xi}_{\text{\tiny{MLE}}})$ are $(\theta_0,\xi_0)$. Under true parameter values, the expectation of the estimating function for (4.10) is
\begin{align*}
&E\Bigl\{\left\{{q}_1(X,Z)-E\left[{q}_1(X,Z)|X;\xi_0\right] \right\} \times \\
&\left\{ (1-R)  E\left({q}_2(X,Y)|R=0,X,Z;{\zeta_0},\theta_0\right)+R{q}_2(X,Y) \right\} \Bigr\} \\
=& E \{ E(\cdot|R=0,X,Z)\times \Pr(R=0|X,Z)\}\\
&+E \{ E(\cdot|R=1,X,Z)\times \Pr(R=1|X,Z)\}\\
=&E \left( \left\{  q_1(X,Z)-E[q_1(X,Z)|X;\xi_0] \right\} E [q_2(X,Y)|X,Z]  \right)\\
=&E \left( \left\{  q_1(X,Z)-E[q_1(X,Z)|X;\xi_0] \right\} E [q_2(X,Y)|X]  \right) \hphantom {------}\text{by {\bf (IV.1)}} \\
=&E \left( \left\{E[q_1(X,Z)|X;\xi_0]-E[[q_1(X,Z)|X;\xi_0] \right\}E [q_2(X,Y)|X] \right) \\
=&0,
\end{align*}
so that $\zeta_0$ is the probability limit of the solution $\hat{\zeta}$ of (4.10). The OR estimator is unbiased since
\begin{align*}
 & E \left\{ RY + (1-R) E(Y|R=0,X,Z;{\zeta_0},\theta_0)\right\} \\
=& E \left\{ E \{ RY + (1-R) E(Y|R=0,X,Z) | R=0, X,Z\} \times \Pr(R=0 | X,Z)\right\} \\
+&E \left\{ E \{ RY + (1-R)E(Y|R=0,X,Z) | R=1,X,Z\} \times \Pr(R=1 | X,Z)\right\} \\
=&E\left\{ E\{Y|R=0,X,Z\} \times \Pr(R=0 |X,Z) \right\} \\
&+E \left\{ E\{Y|R=1,X,Z\} \times \Pr(R=1 |X,Z) \right\} \\
=& E \left\{ E \{Y|X,Z\} \right\} \\
=& E\{Y\}=\phi_0.
\end{align*}
The consistency and asymptotic normality of $\hat{\phi}^{\text{OR}}$ can be established under standard regularity conditions for GMM estimators \citep{newey} . A necessary condition is that the probability of observing the outcome is bounded away from zero (S1).

\noindent \textbf{Proof of Proposition 3}

Under model $\mathcal{M}_{\text{\tiny{IPW}}}$, let $\xi_0$ denote the true value for parametric model $q(z|x; \xi)$ and it is clear that  $\hat{\xi}_{\text{\tiny{MLE}}}$ has a probability limit equal to $\xi_0$. Let superscript asterisks denote possibly misspecified models. Let $\theta^*$ denote the probability limit of estimation under model $f^*(y|R=1,x,z; \theta)$ and let $\rho(X,Z)=\int \boldsymbol{u}(x,y)\frac{\exp\left[-\eta(x,y,z;\zeta)\right]f\left(y|R=1,x,z; \theta \right)}{\int \exp[-\eta(x,y,z)]f\left(y|R=1,x,z;\theta\right) \mathrm{d}\mu(y)} \mathrm{d}\mu(y)$. Then at true parameter values $(\zeta_0,\omega_0)$, 
\begin{align*}
&E\left\{ \boldsymbol{G}^{\text{\tiny{DR}}}\left(R,X,Y,Z; {\zeta}_0,\omega_0,\theta^*,\boldsymbol{u}\right)\middle \vert X,Y,Z \right\} \\
=& \boldsymbol{u}(X,Y)-\rho^*(X,Z;\zeta_0,\theta^*)+\rho^*(X,Z;\zeta_0,\theta^*)=\boldsymbol{u}(X,Y),
\end{align*}
and therefore the estimating function for (4.13), under iterated expectations with respect to $(X,Y,Z)$ at $(\xi_0,\zeta_0,\omega_0)$, is
\begin{align*}
&E \bigg\{ \left[\boldsymbol{v}(X,Z)-E\left(\boldsymbol{v}(X,Z)\middle\vert X\right)\right]\left\{ \boldsymbol{u}(X,Y)\right\} \bigg\}\\
=&E \bigg\{ \left[\boldsymbol{v}(X,Z)-E\left(\boldsymbol{v}(X,Z)\middle\vert X\right)\right]\left\{   E \left(\boldsymbol{u}(X,Y) \middle \vert X,Z\right)\right\} \bigg\} \\
=&E \bigg\{ \left[\boldsymbol{v}(X,Z)-E\left(\boldsymbol{v}(X,Z)\middle\vert X\right)\right]\left\{   E \left(\boldsymbol{u}(X,Y) \middle \vert X\right)\right\} \bigg\} \hphantom {------}\text{by {\bf (IV.1)}} \\
=&E \bigg\{ \left[E\left(\boldsymbol{v}(X,Z)\middle\vert X\right)-E\left(\boldsymbol{v}(X,Z)\middle\vert X\right)\right]\left\{   E \left(\boldsymbol{u}(X,Y) \middle \vert X\right)\right\} \bigg\}\\
=&\boldsymbol{0}.
\end{align*}
In addtion, under iterated expectations with respect to $(X,Y,Z)$, 
\begin{align*}
E \left\{\boldsymbol{G}^{\text{\tiny{DR}}}\left(R,X,Y,Z,\zeta_0,\omega_0,\theta^*,\boldsymbol{u}=Y\right)\right\}=E\{Y\}.
\end{align*}
Under model $\mathcal{M}_{\text{\tiny{OR}}}$, let $\omega^*$ denote the probability limit of estimation under model $P^*(r|Y=0,x,z;\omega)$. Then at true parameter values $(\zeta_0,\theta_0)$,
 \begin{align*}
&E\left\{ \boldsymbol{G}^{\text{\tiny{DR}}}\left(R,X,Y,Z; {\zeta}_0,\omega^*,\theta_0,\boldsymbol{u}\right)\middle \vert X,Z \right\} \\
=&E\left\{  \frac{R}{\pi(\zeta_0,\omega^*)}\{\boldsymbol{u}(X,Y)-\rho(X,Z)\}+\rho(X,Z)\middle \vert X,Z \right\}\\
=&E   \left\{ \frac{R\{1-\pi(\zeta_0,\omega^*)\}}{\pi(\zeta_0,\omega^*)}\{\boldsymbol{u}(X,Y)-\rho(X,Z)\} \middle \vert X,Z \right\}\\
& + E \left\{\rho(X,Z)+R\{\boldsymbol{u}(X,Y)- \rho(X,Z)\} \middle \vert X,Z  \right\}\\
=&E \left\{ R\left\{ \mathrm{e}^{-\{\lambda(X,Z;\omega^*)+\eta(X,Y,Z;\zeta_0)\}} \right\} \left \{\boldsymbol{u}(X,Y)-\rho(X,Z)\right \}  \middle\vert  X,Z \right\}\\
&+ E \left\{\boldsymbol{u}(X,Y)\middle \vert X,Z  \right\}\\
=& E \left\{\boldsymbol{u}(X,Y)\middle \vert X,Z  \right\}. \tag{S3} 
\end{align*}
The estimating function for (4.13), under iterated expectations with respect to $(X,Z)$ at $(\xi_0,\zeta_0,\theta_0)$, is
\begin{align*}
=&E \bigg\{ \left[\boldsymbol{v}(X,Z)-E\left(\boldsymbol{v}(X,Z)\middle\vert X\right)\right]\left\{   E \left(\boldsymbol{u}(X,Y) \middle \vert Z,X\right)\right\} \bigg\} \\
=&E \bigg\{ \left[\boldsymbol{v}(X,Z)-E\left(\boldsymbol{v}(X,Z)\middle\vert X\right)\right]\left\{   E \left(\boldsymbol{u}(X,Y) \middle \vert X\right)\right\} \bigg\} \hphantom {------}\text{by {\bf (IV.1)}} \\
=&E \bigg\{ \left[E\left(\boldsymbol{v}(X,Z)\middle\vert X\right)-E\left(\boldsymbol{v}(X,Z)\middle\vert X\right)\right]\left\{   E \left(\boldsymbol{u}(X,Y) \middle \vert X\right)\right\} \bigg\}\\
=&\boldsymbol{0}.
\end{align*}
In addition, under iterated expectations with respect to $(X,Z)$ and with similar reasoning given in (S3),
\begin{align*}
E \left\{\boldsymbol{G}^{\text{\tiny{DR}}}\left(R,X,Y,Z,\zeta_0,\omega^*,\theta_0,\boldsymbol{u}=Y\right)\right\}=E\{Y\}.
\end{align*}
The consistency and asymptotic normality of $\hat{\phi}^{\text{DR}}$ can be established under standard regularity conditions for GMM estimators \citep{newey} . A necessary condition is that the probability of observing the outcome is bounded away from zero (\ref{pos}).

\newpage
\section*{Results on Local Efficiency}
\setcounter{section}{0}
\setcounter{equation}{0}
\def\theequation{S\arabic{section}.\arabic{equation}}
\def\thesection{S\arabic{section}}

\fontsize{12}{14pt plus.8pt minus .6pt}\selectfont

\setcounter{equation}{0}
Let $\left( L,R\right) =(X,Z,Y,R)$ denote the complete data. Suppose we
observe $O=(R,X,Z,YR).$ Furthermore, assume that $Z$ is a valid missing data
IV, such that (i) $Y$ is independent of $Z$ given $X,$ and (ii) $R$ given $%
(X,Y,Z)$ follows a model logit $\Pr \{R=1|X,Z,Y\}=\alpha _{0}\left(
X,Z\right) +\alpha _{y}\left( Y,X,Z\right) $ with $\alpha _{0}\left(
X,Z\right) $ unrestricted and $\alpha _{y}\left( Y,X,Z\right) $ known, and $%
\alpha _{y}\left( 0,X,Z\right) =0.$ Throughout, we assume that $\Pr
\{R=1|X,Z,Y\}>\sigma >0$ w.p.1 for some constant $\sigma .$ Let $\mathcal{N}%
_{1}$ and $\mathcal{N}_{2}$ denote the tangent space of the full data and
the missing data model respectively, such that $\mathcal{N=N}_{1}$ $\oplus 
\mathcal{N}_{2}$ is the tangent space in the full data model. \citet{rot} established that the observed data tangent space is given by $%
\mathcal{N}^{O}=\overline{\mathcal{N}_{1}^{O}+\mathcal{N}_{2}^{O}},$ where $%
\mathcal{N}_{j}^{O}=\overline{R\left( g\circ \Pi _{j}\right) }$ where $%
R\left( \cdot \right) $ is the range of the operator $g:\Omega ^{\left(
L,R\right) }\rightarrow \Omega ^{\left( O\right) }$ is the conditional
expectation operator $g\left( \cdot \right) =E\left[ \cdot |O\right] ,$ $%
\Omega ^{\left( L,R\right) }$ and $\Omega ^{\left( O\right) }$ are the
spaces of all random functions of $(C,L)$ and $O$ respectively. $\Pi _{j}$
is the Hilbert space projection operator from $\Omega ^{\left( L,R\right) }$ onto $\mathcal{N}_{j}$ and $\overline{\mathcal{S}}$ is the close linear span
of the set $\mathcal{S}$. We wish to characterize the orthocomplement to the
tangent space in the observed data model $\mathcal{N}^{O,\bot }=\mathcal{N}%
_{1}^{O,\bot }\cap \mathcal{N}_{2}^{O,\bot }.$ \citet{rot}
showed that  
\begin{eqnarray*}
\mathcal{N}_{1}^{O,\bot } &=&\left\{ N_{1}^{O,\bot }=Rm(L)/\pi \left(
L\right) +N_{car}:m\left( L\right) \in \mathcal{N}_{1}^{\bot }\text{ and }%
N_{car}\in \mathcal{N}_{car}\text{ }\right\}  \\
&&\text{where } \\
\mathcal{N}_{car} &=&\left\{ N_{car}=(1-R)a\left( O\right) -RE\left[
(1-R)a\left( O\right) |L\right] /\pi \left( L\right) :\text{ for any }%
a\left( O\right) \text{ }\in \Omega ^{\left( O\right) }\right\}. 
\end{eqnarray*}%
Thus we need to characterize $\mathcal{N}_{1}^{\bot }.$ By the exclusion
restriction, all scores of $f\left( L\right) $ may be written as%
\[
\mathcal{N}_{1}=\left\{ s\left( L\right) =s_{1}\left( Y|X\right)
+s_{2}\left( Z|X\right) +s_{3}\left( X\right) :E\left( S_{1}|X\right)
=E\left( S_{2}|X\right) =E\left( S_{3}\right) =0\right\}. 
\]%
Therefore  
\[
\mathcal{N}_{1}^{\bot }=\left\{ C-C^{\dag }:C=c\left( Y,X,Z\right) \text{
arbitrary, }C^{\dag }=E\left[ C|Z,X\right] +E\left[ C|Y,X\right] -E\left[ C|X%
\right] \right\}, 
\]%
a result given by \citet{bickel1998efficient} and \citet{10.2307/27798905}. Therefore, we have that 
$\mathcal{N}_{1}^{O,\bot }$ consists of functions  
\[
R\left\{ C-C^{\dag }\right\} /\pi \left( L\right) +(1-R)a\left( O\right) -RE
\left[ (1-R)a\left( O\right) |L\right] /\pi \left( L\right) 
\]%
for arbitrary functions $C=c(L)$ and $A=a(O).$ Also, \citet{rot} establish that $\mathcal{N}_{2}^{O,\bot }=\left\{ b\left( O\right)
:b\left( O\right) \in \mathcal{N}_{2}^{\bot }\right\} $ and therefore,%
\[
\mathcal{N}^{O,\bot }=\left\{ N_{1}^{O,\bot }\in \mathcal{N}_{1}^{O,\bot }:E%
\left[ N_{2}N_{1}^{O,\bot }\right] =0,N_{2}\in \mathcal{N}_{2}\right\}. 
\]%
Note that $\mathcal{N}_{2}=\left\{ N_{2}=\left( R-\pi \left( L\right)
\right) g(X,Z)\text{ for all }g\right\} $, which leads to the following
result. 

\begin{lemma}
\[
\mathcal{N}^{O,\bot }=\left\{ 
\begin{array}{c}
N_{1}^{O,\bot }\left( a_{c}\right) =R\left\{ C-C^{\dag }\right\} /\pi \left(
L\right) +(1-R)a_{c}\left( O\right) -RE\left[ (1-R)a_{c}\left( O\right) |L%
\right] /\pi \left( L\right) : \\ 
a_{c}=E\left[ C-C^{\dag }|R=0,X,Z\right] 
\end{array}%
\right\} 
\]
\end{lemma}

\begin{proof}
$N_{1}^{O,\bot }\left( a_{c}\right) $ is clearly in $\mathcal{N}_{1}^{O,\bot
},$ it suffices to show that the unique solution to the equation $E\left[
N_{1}^{O,\bot ^{\ast }}N_{2}\right] =0,$ for all $%
N_{2}\in \mathcal{N}_{2}$ is given by $N_{1}^{O,\bot ^{\ast }}=N_{1}^{O,\bot
}\left( a_{c}\right) .$ In this vein  
\begin{eqnarray*}
0 &=&E\left[ N_{1}^{O,\bot ^{\ast }}N_{2}\right]  \\
&=&E\left[ \left\{ 
\begin{array}{c}
R\left\{ C-C^{\dag }\right\} /\pi \left( L\right) +(1-R)a^{\ast }\left(
O\right)  \\ 
-RE\left[ (1-R)a^{\ast }\left( O\right) |L\right] /\pi \left( L\right) 
\end{array}%
\right\} \left( R-\pi \left( L\right) \right) g(X,Z)\right] =0\text{ for all 
}g \\
&\Leftrightarrow &0=E\left[ \left\{ 
\begin{array}{c}
R\left\{ C-C^{\dag }\right\} /\pi \left( L\right) +(1-R)a^{\ast }\left(
O\right)  \\ 
-RE\left[ (1-R)a^{\ast }\left( O\right) |L\right] /\pi \left( L\right) 
\end{array}%
\right\} \left( R-\pi \left( L\right) \right) |X,Z\right]  \\
&\Leftrightarrow &0=E\left[ \left( 1-\pi \left( L\right) \right) \left\{
C-C^{\dag }\right\} |X,Z\right] -E\left[ (1-\pi \left( L\right) )\pi \left(
L\right) a^{\ast }\left( O\right) |X,Z\right]  \\
&&-E\left[ \left( 1-\pi \left( L\right) \right) E\left[ (1-R)a^{\ast }\left(
O\right) |L\right] |X,Z\right]  \\
&\Leftrightarrow &0=E\left[ \left( 1-\pi \left( L\right) \right) \left\{
C-C^{\dag }\right\} |X,Z\right] -E\left[ \left( 1-\pi \left( L\right)
\right) a^{\ast }\left( O\right) |X,Z\right]  \\
&\Leftrightarrow &0=E\left[ \left[ E\left[ \left\{ C-C^{\dag }\right\}
|X,R=0,Z\right] -a^{\ast }\left( O\right) \right] (1-R)|X,Z\right] 
\end{eqnarray*}%
Upon writing $a^{\ast }\left( O\right) =a_{1}^{\ast }\left( L\right)
R+a_{2}^{\ast }\left( X,Z\right) (1-R),$ we have that $a_{2}^{\ast }\left(
X,Z\right) =E\left[ \left\{ C-C^{\dag }\right\} |X,R=0,Z\right] =a_{c},$
proving the result. 
\end{proof}
Therefore the ortho-complement to the tangent space in a model where (i) and
(ii) hold is given by $\mathcal{N}^{O,\bot }$. Next, we consider the goal of estimating a full data functional $\phi =\phi
\left( F_{L}\right)=E(Y) $ in the missing data model given by (i) and (ii). Let $%
IF_{\phi ,1}=Y-\phi$ denote the full data influence function in the nonparametric
model which does not assume (i)$.$ Then, in the model that assumes (i) and
(ii) hold we have that 
\[
\widetilde{\mathcal{N}}_{1}^{\bot }=\left\{ 
\begin{array}{c}
k\cdot IF_{\phi ,1}+C-C^{\dag }:\text{for all constants }k\text{  and } \\ 
C=c\left( Y,X,Z\right) \text{ arbitrary,}C^{\dag }=E\left[ C|Z,X\right] +E%
\left[ C|Y,X\right] -E\left[ C|X\right] 
\end{array}%
\right\} 
\]

Similar to Lemma 1, we get the following set of influence functions for $%
\phi $ in the model given by (i) and (ii)

\begin{lemma}
\[
\widetilde{\mathcal{N}}^{O,\bot }=\left\{ 
\begin{array}{c}
\widetilde{N}_{1}^{O,\bot }\left( a_{c,\phi }\right) =R\left\{ k\cdot
IF_{\phi ,1}+C-C^{\dag }\right\} /\pi \left( L\right)  \\ 
+(1-R)a_{c}\left( O\right) -RE\left[ (1-R)a_{c}\left( O\right) |L\right]
/\pi \left( L\right) : \\ 
a_{c,\phi }=E\left[ k\cdot IF_{\phi ,1}+C-C^{\dag }|R=0,X,Z\right] ,\text{%
for arbitrary }C=c\left( Y,X,Z\right) \text{ and constant }k%
\end{array}%
\text{ }\right\} 
\]
\end{lemma}

The proof is similar to that of Lemma 1. Next, lets suppose that (ii) does
not hold, and instead, we have (iii) a parametric model $\alpha _{y}\left(
Y,X,Z;\gamma \right) $ with unknown p-dimensional parameter $\gamma .$ \ Let 
$F_{t}\left( R,L\right) $ denote the complete data submodel indexed by $t$
such that  $F_{0}\left( R,L\right) =$ $F\left( R,L\right) .$\ Under the
submodel, let $\phi (\gamma _{t},t)$ denote the solution to 
\[
0=E_{t}\left\{ \widetilde{N}_{1}^{O,\bot }\left( a_{c,\phi };\phi (t),\gamma
_{t},t\right) \right\} \text{ for all }t\text{ in the model}
\]%
and therefore 
\begin{eqnarray*}
0 &=&\nabla _{t}E_{t}\left\{ \widetilde{N}_{1}^{O,\bot }\left( a_{c,\phi
};\phi (t),\gamma _{t},t\right) \right\}  \\
&=&E\left\{ \widetilde{N}_{1}^{O,\bot }\left( a_{c,\phi };\phi (\gamma
)\right) S\right\} +E\left\{ \nabla _{t}\widetilde{N}_{1}^{O,\bot }\left(
a_{c,\phi };\phi (t),\gamma _{t},t\right) \right\}  \\
&=&E\left\{ \widetilde{N}_{1}^{O,\bot }\left( a_{c,\phi };\phi (\gamma
)\right) S\right\} +E\left\{ \nabla _{\phi }\widetilde{N}_{1}^{O,\bot
}\left( a_{c,\phi };\phi \right) \right\} \nabla _{t}\phi (t) \\
&&+E\left\{ \nabla _{\gamma }\widetilde{N}_{1}^{O,\bot }\left( a_{c,\phi
};\phi ,\gamma \right) \right\} \nabla _{t}\gamma _{t} \\
&&+E\left\{ \nabla _{t}\widetilde{N}_{1}^{O,\bot }\left( a_{c,\phi };\phi
,\gamma ,t\right) \right\} 
\end{eqnarray*}%
Now since $\widetilde{N}_{1}^{O,\bot }\left( a_{c,\phi };\phi ,\gamma
,t\right) $ is orthogonal to all nuisance parameters in the model where $%
\left( \phi ,\gamma \right) $ is known $E\left\{ \nabla _{t}\widetilde{N}%
_{1}^{O,\bot }\left( a_{c,\phi };\phi ,\gamma ,t\right) \right\} =0,$
therefore, we get 
\[
\nabla _{t}\phi (t)=-E\left\{ \nabla _{\phi }\widetilde{N}_{1}^{O,\bot
}\left( a_{c,\phi };\phi \right) \right\} ^{-1}\times \left( E\left\{ 
\widetilde{N}_{1}^{O,\bot }\left( a_{c,\phi };\phi (\gamma )\right)
S\right\} +E\left\{ \nabla _{\gamma }\widetilde{N}_{1}^{O,\bot }\left(
a_{c,\phi };\phi ,\gamma \right) \right\} \nabla _{t}\gamma _{t}\right) 
\]

Note that by Lemma 1%
\[
\nabla _{t}\gamma _{t}=E\left( N_{1}^{O,\bot }\left( a_{d}\right) S\right) 
\]%
where $N_{1}^{O,\bot }\left( a_{d}\right) $ $\in \mathcal{N}^{O,\bot }$ with 
$a_{d}=E\left[ D-D^{\dag }|R=0,X,Z\right] $ with $D$ an arbitrary
p-dimensional function of $L$ $.$ Therefore, we conclude that 
\begin{eqnarray*}
\nabla _{t}\phi (t) &=&-E\left\{ \nabla _{\phi }\widetilde{N}_{1}^{O,\bot
}\left( a_{c,\phi };\phi \right) \right\} ^{-1} \\
&&\times E\left\{ \left[ \widetilde{N}_{1}^{O,\bot }\left( a_{c,\phi };\phi
(\gamma )\right) +E\left\{ \nabla _{\gamma }\widetilde{N}_{1}^{O,\bot
}\left( a_{c,\phi };\phi ,\gamma \right) \right\} N_{1}^{O,\bot }\left(
a_{c}\right) \right] S\right\} 
\end{eqnarray*}%
proving that the orthocomplement to the nuisance tangent space in the model
given by (i) and (iii) is given by 
\[
\widetilde{N}_{1}^{O,\bot }\left( a_{c,\phi };\phi (\gamma )\right)
+E\left\{ \nabla _{\gamma }\widetilde{N}_{1}^{O,\bot }\left( a_{c,\phi
};\phi ,\gamma \right) \right\} N_{1}^{O,\bot }\left( a_{d}\right) 
\]%
Now, we note that $\widetilde{N}_{1}^{O,\bot }\left( a_{c,\phi };\phi
(\gamma )\right) $ can be written $\widetilde{N}_{1}^{O,\bot }\left(
a_{c};\phi (\gamma )\right) +\widetilde{N}_{1}^{O,\bot }\left( a_{\phi
};\phi (\gamma )\right) $ where $a_{c}=E\left[ C-C^{\dag }|R=0,X,Z\right] $
and $a_{\phi }=E\left[ k\cdot IF_{\phi ,1}|R=0,X,Z\right] $

$a_{c,\phi }=E\left[ k\cdot IF_{\phi ,1}+C-C^{\dag }|R=0,X,Z\right] .$ 

Let 

$M=\widetilde{N}_{1}^{O,\bot }\left( a_{c^{\ast }};\phi \right) =\Pi \left( 
\widetilde{N}_{1}^{O,\bot }\left( a_{\phi };\phi (\gamma )\right) |\left\{ 
\widetilde{N}_{1}^{O,\bot }\left( a_{c};\phi (\gamma )\right) :c\right\}
\right) ,$ 

and let $U=N_{1}^{O,\bot }\left( a_{d^{\ast }}\right) $ denote the efficient
influence function of $\gamma .$ Then we have that the efficient influence
function of $\phi $ is given by 
\[
\widetilde{N}_{1}^{O,\bot }\left( a_{\phi };\phi \right) -M+E\left\{ \nabla
_{\gamma }\left[ \widetilde{N}_{1}^{O,\bot }\left( a_{\phi };\gamma \right)
-M\left( \gamma \right) \right] \right\} U
\]%
since $\widetilde{N}_{1}^{O,\bot }\left( a_{\phi };\phi \right) -M$ is in
the tangent space of the model, and so is $U.$

In the special case where $Z$ and $Y$ are binary, $C-C^{\dag }$ can be
written%
\[
b\left( X\right) \left\{ Y-E\left( Y|X\right) \right\} \left\{ Z-E\left(
Z|X\right) \right\} 
\]%
for some function $b,$ so that 
\begin{eqnarray*}
\widetilde{N}_{1}^{O,\bot }\left( a_{c};\phi (\gamma )\right)  &=&b\left(
X\right) \times \left\{ R\left\{ Y-E\left( Y|X\right) \right\} \left\{
Z-E\left( Z|X\right) \right\} /\pi \left( L\right) +\right.  \\
&&(1-R)E\left[ \left\{ Y-E\left( Y|X\right) \right\} \left\{ Z-E\left(
Z|X\right) \right\} |X,R=0,Z\right]  \\
&&\left. -RE\left[ (1-R)E\left[ \left\{ Y-E\left( Y|X\right) \right\}
\left\{ Z-E\left( Z|X\right) \right\} |X,R=0,Z\right] |L\right] /\pi \left(
L\right) \right\}  \\
&=&b\left( X\right) \times W
\end{eqnarray*}%
Therefore, letting $H=\widetilde{N}_{1}^{O,\bot }\left( a_{\phi };\phi
(\gamma )\right) $

\begin{eqnarray*}
M &=&\Pi \left( \widetilde{N}_{1}^{O,\bot }\left( a_{\phi };\phi (\gamma
)\right) |\left\{ \widetilde{N}_{1}^{O,\bot }\left( a_{c};\phi (\gamma
)\right) :c\right\} \right)  \\
&=&E\left\{ HW|X\right\} E\left\{ W^{2}|X\right\} ^{-1}W
\end{eqnarray*}%
and $U=N_{1}^{O,\bot }\left( a_{d^{\ast }}\right) $ solves%
\[
E\left\{ N_{1}^{O,\bot }\left( a_{d^{\ast }}\right) N_{1}^{O,\bot }\left(
a_{d}\right) \right\} =E\left\{ \nabla _{\gamma }N_{1}^{O,\bot }\left(
a_{d};\gamma \right) \right\} \text{ for all }D.
\]%
one can verify that $N_{1}^{O,\bot }\left( a_{d^{\ast }}\right) =D^{\ast
}\left( X\right) \times W\left( \gamma \right) $ where 
\[
D^{\ast }\left( X\right) =E\left\{ W\left( \gamma \right) ^{\otimes
2}|X\right\} ^{-1}E\left\{ \nabla _{\gamma }W\left( \gamma \right)
|X\right\} 
\]

Next, we illustrate the result by constructing
a locally efficient estimator of $\left( \zeta ,\phi \right) $ in the case
where $Z$ and $Y$ are both binary. In this vein, let $L=(X,Z,Y)$ and define 
\begin{align*}
W &=W\left( \zeta _{0}\right) =R\left\{ Y-E(Y|X)\right\} \left\{
Z-E(Z|X)\right\} /\pi \left( L\right)  \\
&+(1-R)E\left[ \left( Y-E(Y|X)\right) \left\{ Z-E(Z|X)\right\} |X,R=0,Z%
\right]  \\
&-RE\left[ (1-R)E\left[ \left\{ Y-E\left( Y|X\right) \right\} \left\{
Z-E(Z|X\right\} |X,R=0,Z\right] |L\right] /\pi \left( L\right).
\end{align*}
A one-step locally efficient estimator of $\zeta _{0}$ in $\mathcal{M}_{np}$
is given by 
\begin{align}\label{eq:zeta}
\widehat{\zeta }_{EFF}=\widehat{\zeta }_{DR}-\mathbb{P}_{n}\left[ \nabla
_{\zeta }\widehat{ES}_{\zeta }|_{\widehat{\zeta }_{DR}}\right] ^{-1}\widehat{%
ES}_{\zeta },
\end{align}
where $\widehat{ES}_{\zeta }$ is the efficient score $ES_{\zeta }$ of $\zeta 
$ evaluated at the estimated intersection submodel $\mathcal{M}_{IPW}\cap 
\mathcal{M}_{OR}$, where 
\[
ES_{\zeta }=E\left[ W\left( \zeta \right) ^{2}|X\right] ^{-1}E\left[ \nabla
_{\zeta }W\left( \zeta \right) |X\right] W\left( \zeta \right). 
\]
Furthermore, let  $\mathbf{u}^{\ast }\left( X,Y\right) \mathbf{=}Y$ and $%
\widehat{\mathbf{G}}^{DR}$ equal to $\mathbf{G}^{DR}\left( R,X,Y,Z;\zeta ,%
\mathbf{u}^{\ast }\right) \ $ evaluated at the estimated intersection
submodel $\mathcal{M}_{IPW}\cap \mathcal{M}_{OR}$, $\widehat{\zeta }%
_{EFF}$ substituted in for $\zeta$. Then, the efficient estimator of $\phi $
is given by 
\begin{align}\label{eq:phi}
\mathbb{P}_{n}\left[ \widehat{\mathbf{G}}^{DR}-\widehat{E}\left[ W\left( 
\widehat{\zeta }_{EFF}\right) ^{2}|X\right] ^{-1}\widehat{E}\left[ \widehat{%
\mathbf{G}}^{DR}W\left( \widehat{\zeta }_{EFF}\right) |X\right] W\left( 
\widehat{\zeta }_{EFF}\right) \right],
\end{align}
where $\widehat{E}$ is the expectation under the estimated intersection
submodel with $\zeta $ estimated efficiently using $\widehat{\zeta }_{EFF}$.

When both $Z$ and $Y$ contain continuous components, the efficient influence functions for $\phi$ and $\zeta$ are in general not available in closed forms, in the sense that they cannot be explicitly expressed as functions of the true distribution. We adopt the general strategy proposed in \citet{Newey1993419} (see also \citet{10.2307/27798905}) to obtain an approximately locally efficient estimator by taking a basis system $\rho_j(Y, X, Z)$ $(j = 1,...)$ of functions dense in $L_2$, such as tensor products of trigonometric, wavelets or polynomial bases. For approximate efficiency, in practice we let $C(Y,X,Z)=\sum_{j=1}^K h_j\rho_j(Y,X,Z)$ for some finite $K$, where $(h_1,...,h_K)^T \in \Re^K$ are constants.  

We first derive an approximately locally efficient estimator for $\zeta$, with influence function ${IF}_{\tilde{\rho}_K,\zeta}$ where $\tilde{\rho}_K=\left\{\rho_1,...\rho_K\right\}$ is the vector of first K basis functions. Let $\kappa_j=N^{O,\bot }\left( a_{\rho_j};\zeta\right)$ and $\tilde{\kappa}=\left\{\kappa_1,...,\kappa_K \right\}^T$, so that ${IF}_{\tilde{\rho}_K,\zeta}=h^{\ast T}\tilde{\kappa}$ for some $h^{\ast}=(h^{\ast}_1,...,h^{\ast}_K)^T\in \Re^K$. By Theorem 5.3 of \citet{newey}, $E\left\{ h^T \tilde{\kappa} h^{\ast T}\tilde{\kappa}\right\}= E\left\{ h^T \nabla _{\gamma }\tilde{\kappa} \right\} \phantom{-}\forall h\in\Re^K.$ It follows that $h^{\ast}=E\left\{\tilde{\kappa}\tilde{\kappa}^T\right\}^{-1}E\left\{ \nabla _{\gamma }\tilde{\kappa} \right\}$. A one-step approximately efficient estimator of $\zeta _{0}$ in $\mathcal{M}_{np}$
is given by 
\begin{align}\label{eq:zeta}
\widehat{\zeta }_{K}=\widehat{\zeta }_{DR}-\mathbb{P}_{n}\left[ \nabla
_{\zeta }\widehat{IF}_{\tilde{\rho}_K,\zeta},\middle|\,_{\widehat{\zeta }_{DR}}\right] ^{-1}\widehat{IF}_{\tilde{\rho}_K,\zeta}.
\end{align}
$\widehat{\zeta }_{DR}$ is the doubly robust estimate for $\zeta_0$ and 
$$
\widehat{IF}_{\tilde{\rho}_K,\zeta}=\left[\widehat{E}\left\{\tilde{\kappa}\tilde{\kappa}^T\right\}^{-1}\widehat{E}\left\{ \nabla _{\gamma }\tilde{\kappa} \right\}\right]^T\tilde{\kappa},
$$
where $\widehat{E}$ is the expectation under the estimated intersection submodel $\mathcal{M}_{IPW}\cap \mathcal{M}_{OR}$. Under standard regularity conditions, the influence function of the one-step updated estimator $\widehat{\zeta }_{K}$ is asymptotically equivalent to that of the estimator $\tilde{\zeta }_{K}$ which solves $\mathbb{P}_{n}\left\{\widehat{IF}_{\tilde{\rho}_K,\zeta}(\tilde{\zeta }_{K})\right\}=0$ \citep{bickel1998efficient}. In particular, the inverse of the asymptotic variance of $\tilde{\zeta }_{K}$ at the intersection submodel is 
\begin{align*}
\Omega_K &= E\left\{ \nabla _{\gamma }\tilde{\kappa} \vert_{\zeta_0}\right\}^T {E}\left\{\tilde{\kappa}\tilde{\kappa}^T\right\}^{-1}E\left\{ \nabla _{\gamma }\tilde{\kappa} \vert_{\zeta_0}\right\}\\
&=E\left\{ S_{\zeta}\tilde{\kappa}^T \right\} {E}\left\{\tilde{\kappa}\tilde{\kappa}^T\right\}^{-1}E\left\{S_{\zeta}\tilde{\kappa}^T\right\}^T,
\end{align*} 
evaluated at $\zeta=\zeta_0$, and $S_{\zeta}$ is the score vector with respect to $\zeta$. Thus, $\Omega_K$ is the variance of the population least squares regression of $S_{\zeta}$ on the linear span of $\tilde{\kappa}$. Since $\tilde{\rho}_K$ is dense in $L_2$, as the dimension $K \to \infty$ the linear span of $\tilde{\kappa}$ recovers the orthocomplement nuisance tangent space $\mathcal{N}^{O,\bot }$ so that $\Omega_K  \to || \Pi\left(S_{\zeta}|\mathcal{N}^{O,\bot }\right)||^2=\text{var}\left(S_{\zeta, \text{eff}}\right)$, the semiparametric information bound for estimating $\zeta_0$ in the union model $\mathcal{M}_{IPW}\cup \mathcal{M}_{OR}$.

Let $H=\widetilde{N}_{1}^{O,\bot }\left( a_{\phi };\phi(\gamma )\right) $ and $\varrho_j=\widetilde{N}_{1}^{O,\bot }\left(a_{\rho_j};\phi(\gamma )\right) $. Then the unique projection of $H$ onto the linear subspace spanned by $\varrho=\left\{\varrho_1,...,\varrho_K \right\}^T$, i.e. $P=\left\{h^T \varrho\phantom{-}\text{for}\phantom{-}h=(h_1,...,h_K)^T\in \Re^K\right\}$, is given by $$M=h_0^T\varrho,$$ where $h_0^T=E(H\varrho^T)\left\{E(\varrho\varrho^T)\right\}^{-1}$ \citep{tsiatis}. Accordingly, the approximate efficient influence function of $\phi$ is given by
\[
H-M+E\left\{ \nabla
_{\gamma }\left( H-M\right) \right\}{IF}_{\tilde{\rho}_K,\zeta},
\]%
and the approximate efficient estimator of $\phi_0 $
is given by 
\begin{align}\label{eq:phi}
\widehat{\phi}_K=\mathbb{P}_{n}\left[ \widehat{\mathbf{G}}^{DR}-\widehat{E}(H\varrho^T)\left\{\widehat{E}(\varrho\varrho^T)\right\}^{-1}\varrho \right], 
\end{align}
where $%
\widehat{\mathbf{G}}^{DR}$ equal to $\mathbf{G}^{DR}\left( R,X,Y,Z;\widehat{\zeta }_{K} ,%
\mathbf{u}^{\ast }\right) \ $ evaluated at the estimated intersection
submodel $\mathcal{M}_{IPW}\cap \mathcal{M}_{OR}$ with $\mathbf{u}^{\ast }\left( X,Y\right) \mathbf{=}Y$ and $\widehat{E}$ is the expectation also at the estimated intersection model evaluated at $\zeta=\widehat{\zeta }_{K}$. The estimator $\widehat{\phi}_K$ is consistent and asymptotically normal in the semiparametric union model $\mathcal{M}_{IPW}\cup \mathcal{M}_{OR}$; furthermore, analogous to the earlier argument on the semiparametric efficiency of $\widehat{\zeta }_{K}$ as $K \to \infty$, it can be shown that the asymptotic variance of $n^{1/2}(\widehat{\phi}_K -\phi_0)$ nearly attains the semiparametric efficiency bound for the union model at the intersection submodel $\mathcal{M}_{IPW}\cap \mathcal{M}_{OR}$ with $K$ chosen sufficiently large.

\newpage
\section*{R Code for Simulation Study}
\begin{Verbatim}[fontsize=\scriptsize]
rm(list=ls())

#sample size
n = 5000
#number of replications
iter = 1000


library("BB")
library("numDeriv")
set.seed(8)

expit <- function(x) {1/(1+exp(-x)) }

ipw.conv <- numeric(iter)
ipw.par  <-matrix(0,iter,5)
ipw.est  <- numeric(iter)
ipw.var  <- numeric(iter)
ipw.com  <- numeric(iter)
ipw.full <- numeric(iter)
ipw.sb   <- numeric(iter)
ipw.sbvar <- numeric(iter)

imp.sb   <- numeric(iter)
imp.est  <- numeric(iter)
imp.var  <- numeric(iter)
imp.sbvar  <- numeric(iter)

imp.sb.m   <- numeric(iter)
imp.est.m  <- numeric(iter)
imp.var.m  <- numeric(iter)
imp.sbvar.m  <- numeric(iter)

ipw.est.m <- numeric(iter)
ipw.sb.m  <- numeric(iter)
ipw.var.m  <- numeric(iter)
ipw.sbvar.m  <- numeric(iter)

dr.sb.pm  <- numeric(iter)
dr.est.pm <- numeric(iter)
dr.var.pm  <- numeric(iter)
dr.sbvar.pm <- numeric(iter)

dr.sb.bm  <- numeric(iter)
dr.est.bm <- numeric(iter)
dr.var.bm  <- numeric(iter)
dr.sbvar.bm <- numeric(iter)

dr.sb   <- numeric(iter)
dr.est  <- numeric(iter)
dr.var  <- numeric(iter)
dr.sbvar <- numeric(iter)

eff.sb  <- numeric(iter)
eff.est  <- numeric(iter)
eff.est2  <- numeric(iter)

ipw.convm <- numeric(iter)
ipw.parm  <-matrix(0,iter,2)
ipw.estm  <- numeric(iter)
ipw.varm  <- numeric(iter)
ipw.comm  <- numeric(iter)
ipw.sbm <- numeric(iter)
ipw.sbvarm <- numeric(iter)

#true value of E(Y)

true_phi <- 0.4*0.6*expit(1-1.2+1.5)+
	0.6*0.6*expit(1    +1.5)+
	0.4*0.4*expit(1-1.2    )+
	0.6*0.4*expit(1)

for (i in 1:iter) {

x1 <- rbinom(n,1,0.4)
x2 <- rbinom(n,1,0.6)
z  <- rbinom(n,1,expit(0.4+0.9*x1-0.7*x2-0.8*x1*x2))
y  <- rbinom(n,1,expit(1.0-1.2*x1+1.5*x2))

r  <- rbinom(n,1,expit(-1.5+2.5*z+0.8*x1-1.2*x2+1.8*y))
pz.x <- glm(z ~ x1 + x2+x1*x2, family="binomial")

#IPW estimation

ipw <- function(g) {
  h<- rep(0,5)
  h[1]<-sum(r/expit(g[1]+g[2]*z+g[3]*x1+g[4]*x2+g[5]*y)-1)
  h[2]<-sum((r/expit(g[1]+g[2]*z+g[3]*x1+g[4]*x2+g[5]*y)-1)*z)
  h[3]<-sum((r/expit(g[1]+g[2]*z+g[3]*x1+g[4]*x2+g[5]*y)-1)*x1)
  h[4]<-sum((r/expit(g[1]+g[2]*z+g[3]*x1+g[4]*x2+g[5]*y)-1)*x2)
  h[5]<-sum(r*y/expit(g[1]+g[2]*z+g[3]*x1+g[4]*x2+g[5]*y)*(z-pz.x$fit))
  h
}

t1 <- system.time(ans.ipw <-
  BBsolve(par = rep(0,5), fn = ipw,quiet=T))[1]


  ipw.conv[i]<-ans.ipw$conv
  ipw.par[i,]<-ans.ipw$par
  ipw.est[i] <-mean(r*y /expit(ans.ipw$par[1]+ans.ipw$par[2]*z+ans.ipw$par[3]*x1
                    +ans.ipw$par[4]*x2+ans.ipw$par[5]*y))
  ipw.com[i]<-mean(r*y)
  ipw.full[i] <- mean(y)
  ipw.sb[i] <- ans.ipw$par[5]

  #stimate asymptotic variance (stacking estimating functions)

M.ipw <- function(g) {
  h<- rep(0,10)

  #estimating functions for P(Z|X)
  h[1]<-sum(z-expit(g[1]+g[2]*x1+g[3]*x2+g[4]*x1*x2))
  h[2]<-sum(x1*(z-expit(g[1]+g[2]*x1+g[3]*x2+g[4]*x1*x2)))
  h[3]<-sum(x2*(z-expit(g[1]+g[2]*x1+g[3]*x2+g[4]*x1*x2)))
  h[4]<-sum(x1*x2*(z-expit(g[1]+g[2]*x1+g[3]*x2+g[4]*x1*x2)))

  #estimating functions for propensity score
  h[5]<-sum(r/expit(g[5]+g[6]*z+g[7]*x1+g[8]*x2+g[9]*y)-1)
  h[6]<-sum((r/expit(g[5]+g[6]*z+g[7]*x1+g[8]*x2+g[9]*y)-1)*z)
  h[7]<-sum((r/expit(g[5]+g[6]*z+g[7]*x1+g[8]*x2+g[9]*y)-1)*x1)
  h[8]<-sum((r/expit(g[5]+g[6]*z+g[7]*x1+g[8]*x2+g[9]*y)-1)*x2)
  h[9]<-sum(r*y/expit(g[5]+g[6]*z+g[7]*x1+g[8]*x2+g[9]*y)*(z-expit(g[1]+g[2]*x1+g[3]*x2+g[4]*x1*x2)))

  #estimating function for E(Y)
  h[10]<-sum(r*y /expit(g[5]+g[6]*z+g[7]*x1+g[8]*x2+g[9]*y)-g[10])	
  h		
}

  dM <- jacobian(func=M.ipw,x=c(pz.x$coef,ans.ipw$par,ipw.est[i]))/n


mm.ipw <- function(g) {
  rbind(			
  #estimating functions for P(Z|X)
  (z-expit(g[1]+g[2]*x1+g[3]*x2+g[4]*x1*x2)),
  (x1*(z-expit(g[1]+g[2]*x1+g[3]*x2+g[4]*x1*x2))),
  (x2*(z-expit(g[1]+g[2]*x1+g[3]*x2+g[4]*x1*x2))),
  (x1*x2*(z-expit(g[1]+g[2]*x1+g[3]*x2+g[4]*x1*x2))),
  #estimating functions for propensity score
  (r/expit(g[5]+g[6]*z+g[7]*x1+g[8]*x2+g[9]*y)-1),
  ((r/expit(g[5]+g[6]*z+g[7]*x1+g[8]*x2+g[9]*y)-1)*z),
  ((r/expit(g[5]+g[6]*z+g[7]*x1+g[8]*x2+g[9]*y)-1)*x1),
  ((r/expit(g[5]+g[6]*z+g[7]*x1+g[8]*x2+g[9]*y)-1)*x2),
  (r*y/expit(g[5]+g[6]*z+g[7]*x1+g[8]*x2+g[9]*y)*(z-expit(g[1]+g[2]*x1+g[3]*x2+g[4]*x1*x2))),

  #estimating function for E(Y)
  (r*y /expit(g[5]+g[6]*z+g[7]*x1+g[8]*x2+g[9]*y)-g[10])	
)
}

m <- mm.ipw(c(pz.x$coef,ans.ipw$par,ipw.est[i]))
ipw.var[i]<-diag(solve(dM)%*%(m%*%t(m)/n)%*%t(solve(dM))/n)[10]
ipw.sbvar[i] <-diag(solve(dM)%*%(m%*%t(m)/n)%*%t(solve(dM))/n)[9]

#OR estimation
	
#estimate complete-case outcome pdf (saturated model)
x1.cc <- x1[r==1]
x2.cc <- x2[r==1]
z.cc <- z[r==1]
y.cc <- y[r==1]
		
pcc <- glm(y.cc ~ x1.cc+x2.cc+z.cc+x1.cc*x2.cc+x1.cc*z.cc+x2.cc*z.cc+x1.cc*x2.cc*z.cc, family="binomial")

#estimate E[Y|X] by IPW
py.x <- function(g) {
  h <- rep(0,3)
  h[1]<-sum(r/expit(ans.ipw$par[1]+ans.ipw$par[2]*z+ans.ipw$par[3]*x1
                    +ans.ipw$par[4]*x2+ans.ipw$par[5]*y)*(y-expit(g[1]+g[2]*x1+g[3]*x2)))
  h[2]<-sum(r/expit(ans.ipw$par[1]+ans.ipw$par[2]*z+ans.ipw$par[3]*x1
                    +ans.ipw$par[4]*x2+ans.ipw$par[5]*y)*(y-expit(g[1]+g[2]*x1+g[3]*x2))*x1)
  h[3]<-sum(r/expit(ans.ipw$par[1]+ans.ipw$par[2]*z+ans.ipw$par[3]*x1
                    +ans.ipw$par[4]*x2+ans.ipw$par[5]*y)*(y-expit(g[1]+g[2]*x1+g[3]*x2))*x2)
  h
}

t1 <- system.time(ans.pyx <-
  BBsolve(par = c(1,-1.5,0.8), fn = py.x, method=3, control = list(M=50),quiet=T))[1]

py.x.fit <- expit(ans.pyx$par[1]+ans.pyx$par[2]*x1+ans.pyx$par[3]*x2)


#outcome pdf when R=0

p.unobs <- function(b,y,x1,x2,z) {
  prob <- expit(pcc$coef[1]+pcc$coef[2]*x1+pcc$coef[3]*x2+pcc$coef[4]*z
                +pcc$coef[5]*x1*x2+pcc$coef[6]*x1*z+pcc$coef[7]*x2*z+pcc$coef[8]*x1*x2*z)
  denom<- exp(-b)*prob+exp(0)*(1-prob)
  return ( y*(exp(-b)*prob/denom) + (1-y)*(exp(0)*(1-prob)/denom) )
}

p.unobs2 <- function(b,d, y,x1,x2,z) {
  prob <- expit(d[1]+d[2]*x1+d[3]*x2+d[4]*z+d[5]*x1*x2+d[6]*x1*z+d[7]*x2*z+d[8]*x1*x2*z)
  denom<- exp(-b)*prob+exp(0)*(1-prob)
  return ( y*(exp(-b)*prob/denom) + (1-y)*(exp(0)*(1-prob)/denom) )
}

#estimating function (15)
  imp <- function(h) {
  sum( (z-pz.x$fit)*((1-r)*((1)*p.unobs(h,1,x1,x2,z)+(0)*p.unobs(h,0,x1,x2,z) )+r*(y)) )
}
imp.sb[i] <- uniroot(imp,c(-5,8))$root
imp.est[i]<- mean(r*y+(1-r)*(p.unobs(imp.sb[i],1,x1,x2,z)))

M.imp <- function(g) {
  h<- rep(0,14)
  #estimating functions for P(Z|X)
  h[1]<-sum(z-expit(g[1]+g[2]*x1+g[3]*x2+g[4]*x1*x2))
  h[2]<-sum(x1*(z-expit(g[1]+g[2]*x1+g[3]*x2+g[4]*x1*x2)))
  h[3]<-sum(x2*(z-expit(g[1]+g[2]*x1+g[3]*x2+g[4]*x1*x2)))
  h[4]<-sum(x1*x2*(z-expit(g[1]+g[2]*x1+g[3]*x2+g[4]*x1*x2)))
  #estimate outcome density parameters
  h[5]<-sum(r*(y-expit(g[5]+g[6]*x1+g[7]*x2+g[8]*z+g[9]*x1*x2+g[10]*x1*z+g[11]*x2*z+g[12]*x1*x2*z)) )
  h[6]<-sum(r*(y-expit(g[5]+g[6]*x1+g[7]*x2+g[8]*z+g[9]*x1*x2+g[10]*x1*z+g[11]*x2*z+g[12]*x1*x2*z))*x1 )
  h[7]<-sum(r*(y-expit(g[5]+g[6]*x1+g[7]*x2+g[8]*z+g[9]*x1*x2+g[10]*x1*z+g[11]*x2*z+g[12]*x1*x2*z))*x2 )
  h[8]<-sum(r*(y-expit(g[5]+g[6]*x1+g[7]*x2+g[8]*z+g[9]*x1*x2+g[10]*x1*z+g[11]*x2*z+g[12]*x1*x2*z))*z )
  h[9]<-sum(r*(y-expit(g[5]+g[6]*x1+g[7]*x2+g[8]*z+g[9]*x1*x2+g[10]*x1*z+g[11]*x2*z+g[12]*x1*x2*z))*x1*x2 )
  h[10]<-sum(r*(y-expit(g[5]+g[6]*x1+g[7]*x2+g[8]*z+g[9]*x1*x2+g[10]*x1*z+g[11]*x2*z+g[12]*x1*x2*z))*x1*z )
  h[11]<-sum(r*(y-expit(g[5]+g[6]*x1+g[7]*x2+g[8]*z+g[9]*x1*x2+g[10]*x1*z+g[11]*x2*z+g[12]*x1*x2*z))*x2*z )
  h[12]<-sum(r*(y-expit(g[5]+g[6]*x1+g[7]*x2+g[8]*z+g[9]*x1*x2+g[10]*x1*z+g[11]*x2*z+g[12]*x1*x2*z))*x1*x2*z )
  #estimating functions for selection bias
  h[13] <- sum( (z-expit(g[1]+g[2]*x1+g[3]*x2+g[4]*x1*x2))*((1-r)
  *(p.unobs2(g[13],d=c(g[5],g[6],g[7],g[8],g[9],g[10],g[11],g[12]),1,x1,x2,z))+r*y) )
  #estimating function for E(Y)
  h[14]<-sum(r*y+(1-r)*(p.unobs2(g[13],d=c(g[5],g[6],g[7],g[8],g[9],g[10],g[11],g[12]),1,x1,x2,z))-g[14])	
  h		
}

dM <- jacobian(func=M.imp,x=c(pz.x$coef,pcc$coef,imp.sb[i],imp.est[i]))/n

mm.imp <- function(g) {
  rbind(		
  (z-expit(g[1]+g[2]*x1+g[3]*x2+g[4]*x1*x2)),
  (x1*(z-expit(g[1]+g[2]*x1+g[3]*x2+g[4]*x1*x2))),
  (x2*(z-expit(g[1]+g[2]*x1+g[3]*x2+g[4]*x1*x2))),
  (x1*x2*(z-expit(g[1]+g[2]*x1+g[3]*x2+g[4]*x1*x2))),	
  #estimate outcome density parameters
  (r*(y-expit(g[5]+g[6]*x1+g[7]*x2+g[8]*z+g[9]*x1*x2+g[10]*x1*z+g[11]*x2*z+g[12]*x1*x2*z)) ),
  (r*(y-expit(g[5]+g[6]*x1+g[7]*x2+g[8]*z+g[9]*x1*x2+g[10]*x1*z+g[11]*x2*z+g[12]*x1*x2*z))*x1 ),
  (r*(y-expit(g[5]+g[6]*x1+g[7]*x2+g[8]*z+g[9]*x1*x2+g[10]*x1*z+g[11]*x2*z+g[12]*x1*x2*z))*x2 ),
  (r*(y-expit(g[5]+g[6]*x1+g[7]*x2+g[8]*z+g[9]*x1*x2+g[10]*x1*z+g[11]*x2*z+g[12]*x1*x2*z))*z ),
  (r*(y-expit(g[5]+g[6]*x1+g[7]*x2+g[8]*z+g[9]*x1*x2+g[10]*x1*z+g[11]*x2*z+g[12]*x1*x2*z))*x1*x2 ),
  (r*(y-expit(g[5]+g[6]*x1+g[7]*x2+g[8]*z+g[9]*x1*x2+g[10]*x1*z+g[11]*x2*z+g[12]*x1*x2*z))*x1*z ),
  (r*(y-expit(g[5]+g[6]*x1+g[7]*x2+g[8]*z+g[9]*x1*x2+g[10]*x1*z+g[11]*x2*z+g[12]*x1*x2*z))*x2*z ),
  (r*(y-expit(g[5]+g[6]*x1+g[7]*x2+g[8]*z+g[9]*x1*x2+g[10]*x1*z+g[11]*x2*z+g[12]*x1*x2*z))*x1*x2*z ),
  #estimating functions for selection bias
  (z-expit(g[1]+g[2]*x1+g[3]*x2+g[4]*x1*x2))*((1-r)
  *(p.unobs2(g[13],d=c(g[5],g[6],g[7],g[8],g[9],g[10],g[11],g[12]),1,x1,x2,z))+r*y),
  #estimating function for E(Y)
  (r*y+(1-r)*(p.unobs2(g[13],d=c(g[5],g[6],g[7],g[8],g[9],g[10],g[11],g[12]),1,x1,x2,z))-g[14])	
  )
}

m <- mm.imp(g=c(pz.x$coef,pcc$coef,imp.sb[i],imp.est[i]))
imp.var[i]<-diag(solve(dM)%*%(m%*%t(m)/n)%*%t(solve(dM))/n)[14]
imp.sbvar[i] <-diag(solve(dM)%*%(m%*%t(m)/n)%*%t(solve(dM))/n)[13]

##Doubly robust estimation

dr <- function(g) {
  sum(
    (z-pz.x$fit)*(r/expit(ans.ipw$par[1]+ans.ipw$par[2]*z+ans.ipw$par[3]*x1
    +ans.ipw$par[4]*x2+g*y)*(y-p.unobs(g,1,x1,x2,z))
    +p.unobs(g,1,x1,x2,z))
  )
}

dr.sb[i]<-uniroot(dr,c(-5,8))$root

dr.est[i] <-mean(r/expit(ans.ipw$par[1]+ans.ipw$par[2]*z+ans.ipw$par[3]*x1+ans.ipw$par[4]*x2+dr.sb[i]*y)
                 *(y-p.unobs(dr.sb[i],1,x1,x2,z))
                 +p.unobs(dr.sb[i],1,x1,x2,z))

M.dr <- function(g) {
  h<- rep(0,18)

  #estimating functions for P(Z|X)
  h[1]<-sum(z-expit(g[1]+g[2]*x1+g[3]*x2+g[4]*x1*x2))
  h[2]<-sum(x1*(z-expit(g[1]+g[2]*x1+g[3]*x2+g[4]*x1*x2)))
  h[3]<-sum(x2*(z-expit(g[1]+g[2]*x1+g[3]*x2+g[4]*x1*x2)))
  h[4]<-sum(x1*x2*(z-expit(g[1]+g[2]*x1+g[3]*x2+g[4]*x1*x2)))

  #estimating functions for propensity score
  h[5]<-sum(r/expit(g[5]+g[6]*z+g[7]*x1+g[8]*x2+g[9]*y)-1)
  h[6]<-sum((r/expit(g[5]+g[6]*z+g[7]*x1+g[8]*x2+g[9]*y)-1)*z)
  h[7]<-sum((r/expit(g[5]+g[6]*z+g[7]*x1+g[8]*x2+g[9]*y)-1)*x1)
  h[8]<-sum((r/expit(g[5]+g[6]*z+g[7]*x1+g[8]*x2+g[9]*y)-1)*x2)
  h[9]<-sum((r/expit(g[5]+g[6]*z+g[7]*x1+g[8]*x2+g[9]*y)
             *(y-p.unobs2(g[9],d=c(g[10],g[11],g[12],g[13],g[14],g[15],g[16],g[17]),1,x1,x2,z))
             +p.unobs2(g[9],d=c(g[10],g[11],g[12],g[13],g[14],g[15],g[16],g[17]),1,x1,x2,z))
             *(z-expit(g[1]+g[2]*x1+g[3]*x2+g[4]*x1*x2)))

  #estimate outcome density parameters
  h[10]<-sum(r*(y-expit(g[10]+g[11]*x1+g[12]*x2+g[13]*z+g[14]*x1*x2+g[15]*x1*z+g[16]*x2*z+g[17]*x1*x2*z)) )
  h[11]<-sum(r*(y-expit(g[10]+g[11]*x1+g[12]*x2+g[13]*z+g[14]*x1*x2+g[15]*x1*z+g[16]*x2*z+g[17]*x1*x2*z))*x1 )
  h[12]<-sum(r*(y-expit(g[10]+g[11]*x1+g[12]*x2+g[13]*z+g[14]*x1*x2+g[15]*x1*z+g[16]*x2*z+g[17]*x1*x2*z))*x2 )
  h[13]<-sum(r*(y-expit(g[10]+g[11]*x1+g[12]*x2+g[13]*z+g[14]*x1*x2+g[15]*x1*z+g[16]*x2*z+g[17]*x1*x2*z))*z )
  h[14]<-sum(r*(y-expit(g[10]+g[11]*x1+g[12]*x2+g[13]*z+g[14]*x1*x2+g[15]*x1*z+g[16]*x2*z+g[17]*x1*x2*z))*x1*x2 )
  h[15]<-sum(r*(y-expit(g[10]+g[11]*x1+g[12]*x2+g[13]*z+g[14]*x1*x2+g[15]*x1*z+g[16]*x2*z+g[17]*x1*x2*z))*x1*z )
  h[16]<-sum(r*(y-expit(g[10]+g[11]*x1+g[12]*x2+g[13]*z+g[14]*x1*x2+g[15]*x1*z+g[16]*x2*z+g[17]*x1*x2*z))*x2*z )
  h[17]<-sum(r*(y-expit(g[10]+g[11]*x1+g[12]*x2+g[13]*z+g[14]*x1*x2+g[15]*x1*z+g[16]*x2*z+g[17]*x1*x2*z))*x1*x2*z )

  #estimating function for E(Y)
  h[18]<-sum((r/expit(g[5]+g[6]*z+g[7]*x1+g[8]*x2+g[9]*y)
             *(y-p.unobs2(g[9],d=c(g[10],g[11],g[12],g[13],g[14],g[15],g[16],g[17]),1,x1,x2,z))
             +p.unobs2(g[9],d=c(g[10],g[11],g[12],g[13],g[14],g[15],g[16],g[17]),1,x1,x2,z))-g[18])	
  h		
}

  dM <- jacobian(func=M.dr,x=c(pz.x$coef,ans.ipw$par[1:4],dr.sb[i],pcc$coef,dr.est[i]))/n


mm.dr<- function(g) {
rbind(			
  #estimating functions for P(Z|X)
  (z-expit(g[1]+g[2]*x1+g[3]*x2+g[4]*x1*x2)),
  (x1*(z-expit(g[1]+g[2]*x1+g[3]*x2+g[4]*x1*x2))),
  (x2*(z-expit(g[1]+g[2]*x1+g[3]*x2+g[4]*x1*x2))),
  (x1*x2*(z-expit(g[1]+g[2]*x1+g[3]*x2+g[4]*x1*x2))),

  #estimating functions for propensity score
  (r/expit(g[5]+g[6]*z+g[7]*x1+g[8]*x2+g[9]*y)-1),
  ((r/expit(g[5]+g[6]*z+g[7]*x1+g[8]*x2+g[9]*y)-1)*z),
  ((r/expit(g[5]+g[6]*z+g[7]*x1+g[8]*x2+g[9]*y)-1)*x1),
  ((r/expit(g[5]+g[6]*z+g[7]*x1+g[8]*x2+g[9]*y)-1)*x2),
  ((r/expit(g[5]+g[6]*z+g[7]*x1+g[8]*x2+g[9]*y)
    *(y-p.unobs2(g[9],d=c(g[10],g[11],g[12],g[13],g[14],g[15],g[16],g[17]),1,x1,x2,z))
    +p.unobs2(g[9],d=c(g[10],g[11],g[12],g[13],g[14],g[15],g[16],g[17]),1,x1,x2,z))
    *(z-expit(g[1]+g[2]*x1+g[3]*x2+g[4]*x1*x2))),

  #estimate outcome density parameters
  (r*(y-expit(g[10]+g[11]*x1+g[12]*x2+g[13]*z+g[14]*x1*x2+g[15]*x1*z+g[16]*x2*z+g[17]*x1*x2*z)) ),
  (r*(y-expit(g[10]+g[11]*x1+g[12]*x2+g[13]*z+g[14]*x1*x2+g[15]*x1*z+g[16]*x2*z+g[17]*x1*x2*z))*x1 ),
  (r*(y-expit(g[10]+g[11]*x1+g[12]*x2+g[13]*z+g[14]*x1*x2+g[15]*x1*z+g[16]*x2*z+g[17]*x1*x2*z))*x2 ),
  (r*(y-expit(g[10]+g[11]*x1+g[12]*x2+g[13]*z+g[14]*x1*x2+g[15]*x1*z+g[16]*x2*z+g[17]*x1*x2*z))*z ),
  (r*(y-expit(g[10]+g[11]*x1+g[12]*x2+g[13]*z+g[14]*x1*x2+g[15]*x1*z+g[16]*x2*z+g[17]*x1*x2*z))*x1*x2 ),
  (r*(y-expit(g[10]+g[11]*x1+g[12]*x2+g[13]*z+g[14]*x1*x2+g[15]*x1*z+g[16]*x2*z+g[17]*x1*x2*z))*x1*z ),
  (r*(y-expit(g[10]+g[11]*x1+g[12]*x2+g[13]*z+g[14]*x1*x2+g[15]*x1*z+g[16]*x2*z+g[17]*x1*x2*z))*x2*z ),
  (r*(y-expit(g[10]+g[11]*x1+g[12]*x2+g[13]*z+g[14]*x1*x2+g[15]*x1*z+g[16]*x2*z+g[17]*x1*x2*z))*x1*x2*z ),

  #estimating function for E(Y)
  ((r/expit(g[5]+g[6]*z+g[7]*x1+g[8]*x2+g[9]*y)
    *(y-p.unobs2(g[9],d=c(g[10],g[11],g[12],g[13],g[14],g[15],g[16],g[17]),1,x1,x2,z))
    +p.unobs2(g[9],d=c(g[10],g[11],g[12],g[13],g[14],g[15],g[16],g[17]),1,x1,x2,z))-g[18])	

  )
}
  m <- mm.dr(c(pz.x$coef,ans.ipw$par[1:4],dr.sb[i],pcc$coef,dr.est[i]))
  dr.var[i]<-diag(solve(dM)%*%(m%*%t(m)/n)%*%t(solve(dM))/n)[18]
  dr.sbvar[i] <-diag(solve(dM)%*%(m%*%t(m)/n)%*%t(solve(dM))/n)[9]
	
  ##misspecified propensity score model
		
  ipw.m <- function(g) {
  h<- rep(0,5)
  h[1]<-sum(r/expit(g[1]+g[2]*z+g[3]*x1+g[4]*x1*z+g[5]*y)-1)
  h[2]<-sum((r/expit(g[1]+g[2]*z+g[3]*x1+g[4]*x1*z+g[5]*y)-1)*z)
  h[3]<-sum((r/expit(g[1]+g[2]*z+g[3]*x1+g[4]*x1*z+g[5]*y)-1)*x1)
  h[4]<-sum((r/expit(g[1]+g[2]*z+g[3]*x1+g[4]*x1*z+g[5]*y)-1)*x1*z)
  h[5]<-sum(r*y/expit(g[1]+g[2]*z+g[3]*x1+g[4]*x1*z+g[5]*y)*(z-pz.x$fit))
  return(sum(h*h))
}
  ans.ipw.m <- optim( rep(0,5), ipw.m, gr = NULL)$par

  ipw.est.m[i] <-mean(r*y /(expit(ans.ipw.m[1]+ans.ipw.m[2]*z
                      +ans.ipw.m[3]*x1+ans.ipw.m[4]*x1*z+ans.ipw.m[5]*y))) 
  ipw.sb.m[i] <- ans.ipw.m[5]

  #stimate asymptotic variance (stacking estimating functions)

  M.ipw.m <- function(g) {
  h<- rep(0,10)

  #estimating functions for P(Z|X)
  h[1]<-sum(z-expit(g[1]+g[2]*x1+g[3]*x2+g[4]*x1*x2))
  h[2]<-sum(x1*(z-expit(g[1]+g[2]*x1+g[3]*x2+g[4]*x1*x2)))
  h[3]<-sum(x2*(z-expit(g[1]+g[2]*x1+g[3]*x2+g[4]*x1*x2)))
  h[4]<-sum(x1*x2*(z-expit(g[1]+g[2]*x1+g[3]*x2+g[4]*x1*x2)))

  #estimating functions for propensity score
  h[5]<-sum(r/expit(g[5]+g[6]*z+g[7]*x1+g[8]*x1*z+g[9]*y)-1)
  h[6]<-sum((r/expit(g[5]+g[6]*z+g[7]*x1+g[8]*x1*z+g[9]*y)-1)*z)
  h[7]<-sum((r/expit(g[5]+g[6]*z+g[7]*x1+g[8]*x1*z+g[9]*y)-1)*x1)
  h[8]<-sum((r/expit(g[5]+g[6]*z+g[7]*x1+g[8]*x1*z+g[9]*y)-1)*x1*z)
  h[9]<-sum(r*y/expit(g[5]+g[6]*z+g[7]*x1+g[8]*x1*z+g[9]*y)*(z-expit(g[1]+g[2]*x1+g[3]*x2+g[4]*x1*x2)))

  #estimating function for E(Y)
  h[10]<-sum(r*y /expit(g[5]+g[6]*z+g[7]*x1+g[8]*x1*z+g[9]*y)-g[10])	
  h		
}

dM <- jacobian(func=M.ipw.m,x=c(pz.x$coef,ans.ipw.m,ipw.est.m[i]))/n


mm.ipw.m <- function(g) {
  rbind(			
  #estimating functions for P(Z|X)
  (z-expit(g[1]+g[2]*x1+g[3]*x2+g[4]*x1*x2)),
  (x1*(z-expit(g[1]+g[2]*x1+g[3]*x2+g[4]*x1*x2))),
  (x2*(z-expit(g[1]+g[2]*x1+g[3]*x2+g[4]*x1*x2))),
  (x1*x2*(z-expit(g[1]+g[2]*x1+g[3]*x2+g[4]*x1*x2))),

  #estimating functions for propensity score
  (r/expit(g[5]+g[6]*z+g[7]*x1+g[8]*x1*z+g[9]*y)-1),
  ((r/expit(g[5]+g[6]*z+g[7]*x1+g[8]*x1*z+g[9]*y)-1)*z),
  ((r/expit(g[5]+g[6]*z+g[7]*x1+g[8]*x1*z+g[9]*y)-1)*x1),
  ((r/expit(g[5]+g[6]*z+g[7]*x1+g[8]*x1*z+g[9]*y)-1)*x1*z),
  (r*y/expit(g[5]+g[6]*z+g[7]*x1+g[8]*x1*z+g[9]*y)*(z-expit(g[1]+g[2]*x1+g[3]*x2+g[4]*x1*x2))),

  #estimating function for E(Y)
  (r*y /expit(g[5]+g[6]*z+g[7]*x1+g[8]*x1*z+g[9]*y)-g[10])
  )
}

  m <- mm.ipw.m(c(pz.x$coef,ans.ipw.m,ipw.est.m[i]))
  ipw.var.m[i]<-diag(solve(dM)%*%(m%*%t(m)/n)%*%t(solve(dM))/n)[10]
  ipw.sbvar.m[i] <-diag(solve(dM)%*%(m%*%t(m)/n)%*%t(solve(dM))/n)[9]

dr.pm <- function(g) {
  h<- rep(0,5)
  h[1]<-sum(r/expit(g[1]+g[2]*z+g[3]*x1+g[4]*x1*z+g[5]*y)-1)
  h[2]<-sum((r/expit(g[1]+g[2]*z+g[3]*x1+g[4]*x1*z+g[5]*y)-1)*z)
  h[3]<-sum((r/expit(g[1]+g[2]*z+g[3]*x1+g[4]*x1*z+g[5]*y)-1)*x1)
  h[4]<-sum((r/expit(g[1]+g[2]*z+g[3]*x1+g[4]*x1*z+g[5]*y)-1)*x1*z)

			
  h[5]<-sum(
           (z-pz.x$fit)*(r/(expit(g[1]+g[2]*z+g[3]*x1+g[4]*x1*z+g[5]*y))*(y-p.unobs(g[5],1,x1,x2,z))
           +p.unobs(g[5],1,x1,x2,z))
			)
}
  t1 <- system.time(ans.dr.pm <-
  BBsolve(par = rep(0,5), fn = dr.pm, quiet=T))[1]

  dr.sb.pm[i]<-ans.dr.pm$par[5]

  dr.est.pm[i] <- mean(r/(expit(ans.dr.pm$par[1]+ans.dr.pm$par[2]*z+ans.dr.pm$par[3]*x1
                       +ans.dr.pm$par[4]*x1*z+dr.sb.pm[i]*y))*(y-p.unobs(dr.sb.pm[i],1,x1,x2,z))
                       +p.unobs(dr.sb.pm[i],1,x1,x2,z))
	
M.dr.pm <- function(g) {
  h<- rep(0,18)

  #estimating functions for P(Z|X)
  h[1]<-sum(z-expit(g[1]+g[2]*x1+g[3]*x2+g[4]*x1*x2))
  h[2]<-sum(x1*(z-expit(g[1]+g[2]*x1+g[3]*x2+g[4]*x1*x2)))
  h[3]<-sum(x2*(z-expit(g[1]+g[2]*x1+g[3]*x2+g[4]*x1*x2)))
  h[4]<-sum(x1*x2*(z-expit(g[1]+g[2]*x1+g[3]*x2+g[4]*x1*x2)))

  #estimating functions for propensity score
  h[5]<-sum(r/expit(g[5]+g[6]*z+g[7]*x1+g[8]*x1*z+g[9]*y)-1)
  h[6]<-sum((r/expit(g[5]+g[6]*z+g[7]*x1+g[8]*x1*z+g[9]*y)-1)*z)
  h[7]<-sum((r/expit(g[5]+g[6]*z+g[7]*x1+g[8]*x1*z+g[9]*y)-1)*x1)
  h[8]<-sum((r/expit(g[5]+g[6]*z+g[7]*x1+g[8]*x1*z+g[9]*y)-1)*x2)
  h[9]<-sum((r/expit(g[5]+g[6]*z+g[7]*x1+g[8]*x1*z+g[9]*y)
             *(y-p.unobs2(g[9],d=c(g[10],g[11],g[12],g[13],g[14],g[15],g[16],g[17]),1,x1,x2,z))
             +p.unobs2(g[9],d=c(g[10],g[11],g[12],g[13],g[14],g[15],g[16],g[17]),1,x1,x2,z))
             *(z-expit(g[1]+g[2]*x1+g[3]*x2+g[4]*x1*x2)))

  #estimate outcome density parameters
  h[10]<-sum(r*(y-expit(g[10]+g[11]*x1+g[12]*x2+g[13]*z+g[14]*x1*x2+g[15]*x1*z+g[16]*x2*z+g[17]*x1*x2*z)) )
  h[11]<-sum(r*(y-expit(g[10]+g[11]*x1+g[12]*x2+g[13]*z+g[14]*x1*x2+g[15]*x1*z+g[16]*x2*z+g[17]*x1*x2*z))*x1 )
  h[12]<-sum(r*(y-expit(g[10]+g[11]*x1+g[12]*x2+g[13]*z+g[14]*x1*x2+g[15]*x1*z+g[16]*x2*z+g[17]*x1*x2*z))*x2 )
  h[13]<-sum(r*(y-expit(g[10]+g[11]*x1+g[12]*x2+g[13]*z+g[14]*x1*x2+g[15]*x1*z+g[16]*x2*z+g[17]*x1*x2*z))*z )
  h[14]<-sum(r*(y-expit(g[10]+g[11]*x1+g[12]*x2+g[13]*z+g[14]*x1*x2+g[15]*x1*z+g[16]*x2*z+g[17]*x1*x2*z))*x1*x2 )
  h[15]<-sum(r*(y-expit(g[10]+g[11]*x1+g[12]*x2+g[13]*z+g[14]*x1*x2+g[15]*x1*z+g[16]*x2*z+g[17]*x1*x2*z))*x1*z )
  h[16]<-sum(r*(y-expit(g[10]+g[11]*x1+g[12]*x2+g[13]*z+g[14]*x1*x2+g[15]*x1*z+g[16]*x2*z+g[17]*x1*x2*z))*x2*z )
  h[17]<-sum(r*(y-expit(g[10]+g[11]*x1+g[12]*x2+g[13]*z+g[14]*x1*x2+g[15]*x1*z+g[16]*x2*z+g[17]*x1*x2*z))*x1*x2*z )

  #estimating function for E(Y)
  h[18]<-sum((r/expit(g[5]+g[6]*z+g[7]*x1+g[8]*x1*z+g[9]*y)
            *(y-p.unobs2(g[9],d=c(g[10],g[11],g[12],g[13],g[14],g[15],g[16],g[17]),1,x1,x2,z))
            +p.unobs2(g[9],d=c(g[10],g[11],g[12],g[13],g[14],g[15],g[16],g[17]),1,x1,x2,z))-g[18])	
  h		
}

  dM <- jacobian(func=M.dr.pm,x=c(pz.x$coef,ans.dr.pm$par,pcc$coef,dr.est.pm[i]))/n


mm.dr.pm<- function(g) {
  rbind(			
  (z-expit(g[1]+g[2]*x1+g[3]*x2+g[4]*x1*x2)),
  (x1*(z-expit(g[1]+g[2]*x1+g[3]*x2+g[4]*x1*x2))),
  (x2*(z-expit(g[1]+g[2]*x1+g[3]*x2+g[4]*x1*x2))),
  (x1*x2*(z-expit(g[1]+g[2]*x1+g[3]*x2+g[4]*x1*x2))),

  #estimating functions for propensity score
  (r/expit(g[5]+g[6]*z+g[7]*x1+g[8]*x1*z+g[9]*y)-1),
  ((r/expit(g[5]+g[6]*z+g[7]*x1+g[8]*x1*z+g[9]*y)-1)*z),
  ((r/expit(g[5]+g[6]*z+g[7]*x1+g[8]*x1*z+g[9]*y)-1)*x1),
  ((r/expit(g[5]+g[6]*z+g[7]*x1+g[8]*x1*z+g[9]*y)-1)*x2),
  ((r/expit(g[5]+g[6]*z+g[7]*x1+g[8]*x1*z+g[9]*y)
    *(y-p.unobs2(g[9],d=c(g[10],g[11],g[12],g[13],g[14],g[15],g[16],g[17]),1,x1,x2,z))
    +p.unobs2(g[9],d=c(g[10],g[11],g[12],g[13],g[14],g[15],g[16],g[17]),1,x1,x2,z))
    *(z-expit(g[1]+g[2]*x1+g[3]*x2+g[4]*x1*x2))),

  #estimate outcome density parameters
  (r*(y-expit(g[10]+g[11]*x1+g[12]*x2+g[13]*z+g[14]*x1*x2+g[15]*x1*z+g[16]*x2*z+g[17]*x1*x2*z)) ),
  (r*(y-expit(g[10]+g[11]*x1+g[12]*x2+g[13]*z+g[14]*x1*x2+g[15]*x1*z+g[16]*x2*z+g[17]*x1*x2*z))*x1 ),
  (r*(y-expit(g[10]+g[11]*x1+g[12]*x2+g[13]*z+g[14]*x1*x2+g[15]*x1*z+g[16]*x2*z+g[17]*x1*x2*z))*x2 ),
  (r*(y-expit(g[10]+g[11]*x1+g[12]*x2+g[13]*z+g[14]*x1*x2+g[15]*x1*z+g[16]*x2*z+g[17]*x1*x2*z))*z ),
  (r*(y-expit(g[10]+g[11]*x1+g[12]*x2+g[13]*z+g[14]*x1*x2+g[15]*x1*z+g[16]*x2*z+g[17]*x1*x2*z))*x1*x2 ),
  (r*(y-expit(g[10]+g[11]*x1+g[12]*x2+g[13]*z+g[14]*x1*x2+g[15]*x1*z+g[16]*x2*z+g[17]*x1*x2*z))*x1*z ),
  (r*(y-expit(g[10]+g[11]*x1+g[12]*x2+g[13]*z+g[14]*x1*x2+g[15]*x1*z+g[16]*x2*z+g[17]*x1*x2*z))*x2*z ),
  (r*(y-expit(g[10]+g[11]*x1+g[12]*x2+g[13]*z+g[14]*x1*x2+g[15]*x1*z+g[16]*x2*z+g[17]*x1*x2*z))*x1*x2*z ),

  #estimating function for E(Y)
  ((r/expit(g[5]+g[6]*z+g[7]*x1+g[8]*x1*z+g[9]*y)
  *(y-p.unobs2(g[9],d=c(g[10],g[11],g[12],g[13],g[14],g[15],g[16],g[17]),1,x1,x2,z))
  +p.unobs2(g[9],d=c(g[10],g[11],g[12],g[13],g[14],g[15],g[16],g[17]),1,x1,x2,z))-g[18])			
  )
}

m <- mm.dr.pm(c(pz.x$coef,ans.dr.pm$par,pcc$coef,dr.est.pm[i]))
dr.var.pm[i]<-diag(solve(dM)%*%(m%*%t(m)/n)%*%t(solve(dM))/n)[18]
dr.sbvar.pm[i] <-diag(solve(dM)%*%(m%*%t(m)/n)%*%t(solve(dM))/n)[9]

###########misspecified outcome regression####################

pcc.m <- glm(y.cc ~ x1.cc, family="binomial")

p.unobs.m <- function(b,y,x1,x2,z) {
  prob <- expit(pcc.m$coef[1]+pcc.m$coef[2]*x1)
  denom<- exp(-b)*prob+exp(0)*(1-prob)
  return ( y*(exp(-b)*prob/denom) + (1-y)*(exp(0)*(1-prob)/denom) )
}

p.unobs.m.2 <- function(b,d,y,x1,x2,z) {
  prob <- expit(d[1]+d[2]*x1)
  denom<- exp(-b)*prob+exp(0)*(1-prob)
  return ( y*(exp(-b)*prob/denom) + (1-y)*(exp(0)*(1-prob)/denom) )
}


imp.m <- function(h) {
  sum( (z-pz.x$fit)*((1-r)*((1)*p.unobs.m(h,1,x1,x2,z)+(0)*p.unobs.m(h,0,x1,x2,z) )+r*(y)) )
}
imp.sb.m[i] <- uniroot(imp.m,c(-5,8))$root
imp.est.m[i]<- mean(r*y+(1-r)*(p.unobs.m(imp.sb.m[i],1,x1,x2,z)))

M.imp.m <- function(g) {
  h<- rep(0,8)
  #estimating functions for P(Z|X)
  h[1]<-sum(z-expit(g[1]+g[2]*x1+g[3]*x2+g[4]*x1*x2))
  h[2]<-sum(x1*(z-expit(g[1]+g[2]*x1+g[3]*x2+g[4]*x1*x2)))
  h[3]<-sum(x2*(z-expit(g[1]+g[2]*x1+g[3]*x2+g[4]*x1*x2)))
  h[4]<-sum(x1*x2*(z-expit(g[1]+g[2]*x1+g[3]*x2+g[4]*x1*x2)))
  #estimate outcome density parameters
  h[5]<-sum(r*(y-expit(g[5]+g[6]*x1)) )
  h[6]<-sum(r*(y-expit(g[5]+g[6]*x1))*x1 )

  #estimating functions for selection bias
  h[7] <- sum( (z-expit(g[1]+g[2]*x1+g[3]*x2+g[4]*x1*x2))*((1-r)
                *(p.unobs.m.2(g[7],d=c(g[5],g[6]),1,x1,x2,z))+r*y) )
  #estimating function for E(Y)
  h[8]<-sum(r*y+(1-r)*(p.unobs.m.2(g[7],d=c(g[5],g[6]),1,x1,x2,z))-g[8])	
  h		
}

dM <- jacobian(func=M.imp.m,x=c(pz.x$coef,pcc.m$coef,imp.sb.m[i],imp.est.m[i]))/n

mm.imp.m <- function(g) {
  rbind(		
  #estimating functions for P(Z|X)
  (z-expit(g[1]+g[2]*x1+g[3]*x2+g[4]*x1*x2)),
  (x1*(z-expit(g[1]+g[2]*x1+g[3]*x2+g[4]*x1*x2))),
  (x2*(z-expit(g[1]+g[2]*x1+g[3]*x2+g[4]*x1*x2))),
  (x1*x2*(z-expit(g[1]+g[2]*x1+g[3]*x2+g[4]*x1*x2))),
  #estimate outcome density parameters
  (r*(y-expit(g[5]+g[6]*x1)) ),
  (r*(y-expit(g[5]+g[6]*x1))*x1 ),

  #estimating functions for selection bias
  ( (z-expit(g[1]+g[2]*x1+g[3]*x2+g[4]*x1*x2))*((1-r)
  *(p.unobs.m.2(g[7],d=c(g[5],g[6]),1,x1,x2,z))+r*y) ),
  #estimating function for E(Y)
  (r*y+(1-r)*(p.unobs.m.2(g[7],d=c(g[5],g[6]),1,x1,x2,z))-g[8])	
  )
}

m <- mm.imp.m(c(pz.x$coef,pcc.m$coef,imp.sb.m[i],imp.est.m[i]))
imp.var.m[i]<-diag(solve(dM)%*%(m%*%t(m)/n)%*%t(solve(dM))/n)[8]
imp.sbvar.m[i] <-diag(solve(dM)%*%(m%*%t(m)/n)%*%t(solve(dM))/n)[7]

dr.bm <- function(g) {
  sum(
     (z-pz.x$fit)*(r/expit(ans.ipw$par[1]+ans.ipw$par[2]*z
     +ans.ipw$par[3]*x1+ans.ipw$par[4]*x2+g*y)*(y-p.unobs.m(g,1,x1,x2,z))
     +p.unobs.m(g,1,x1,x2,z))
  )
}

dr.sb.bm[i]<-uniroot(dr.bm,c(-5,8))$root

dr.est.bm[i] <-mean(r/expit(ans.ipw$par[1]+ans.ipw$par[2]*z+ans.ipw$par[3]*x1
                    +ans.ipw$par[4]*x2+dr.sb.bm[i]*y)*(y-p.unobs.m(dr.sb.bm[i],1,x1,x2,z))
                    +p.unobs.m(dr.sb.bm[i],1,x1,x2,z))

M.dr.bm <- function(g) {
  h<- rep(0,12)

  #estimating functions for P(Z|X)
  h[1]<-sum(z-expit(g[1]+g[2]*x1+g[3]*x2+g[4]*x1*x2))
  h[2]<-sum(x1*(z-expit(g[1]+g[2]*x1+g[3]*x2+g[4]*x1*x2)))
  h[3]<-sum(x2*(z-expit(g[1]+g[2]*x1+g[3]*x2+g[4]*x1*x2)))
  h[4]<-sum(x1*x2*(z-expit(g[1]+g[2]*x1+g[3]*x2+g[4]*x1*x2)))

  #estimating functions for propensity score
  h[5]<-sum(r/expit(g[5]+g[6]*z+g[7]*x1+g[8]*x2+g[9]*y)-1)
  h[6]<-sum((r/expit(g[5]+g[6]*z+g[7]*x1+g[8]*x2+g[9]*y)-1)*z)
  h[7]<-sum((r/expit(g[5]+g[6]*z+g[7]*x1+g[8]*x2+g[9]*y)-1)*x1)
  h[8]<-sum((r/expit(g[5]+g[6]*z+g[7]*x1+g[8]*x2+g[9]*y)-1)*x2)
  h[9]<-sum((r/expit(g[5]+g[6]*z+g[7]*x1+g[8]*x2+g[9]*y)*(y-p.unobs.m.2(g[9],d=c(g[10],g[11]),1,x1,x2,z))
  +p.unobs.m.2(g[9],d=c(g[10],g[11]),1,x1,x2,z))
  *(z-expit(g[1]+g[2]*x1+g[3]*x2+g[4]*x1*x2)))

  #estimate outcome density parameters
  h[10]<-sum(r*(y-expit(g[10]+g[11]*x1)) )
  h[11]<-sum(r*(y-expit(g[10]+g[11]*x1))*x1 )

  #estimating function for E(Y)
  h[12]<-sum((r/expit(g[5]+g[6]*z+g[7]*x1+g[8]*x2+g[9]*y)
             *(y-p.unobs.m.2(g[9],d=c(g[10],g[11]),1,x1,x2,z))
             +p.unobs.m.2(g[9],d=c(g[10],g[11]),1,x1,x2,z))-g[12])	
  h		
}

dM <- jacobian(func=M.dr.bm,x=c(pz.x$coef,ans.ipw$par[1:4],dr.sb.bm[i],pcc.m$coef,dr.est.bm[i]))/n


mm.dr.bm<- function(g) {
  rbind(			
  #estimating functions for P(Z|X)
  (z-expit(g[1]+g[2]*x1+g[3]*x2+g[4]*x1*x2)),
  (x1*(z-expit(g[1]+g[2]*x1+g[3]*x2+g[4]*x1*x2))),
  (x2*(z-expit(g[1]+g[2]*x1+g[3]*x2+g[4]*x1*x2))),
  (x1*x2*(z-expit(g[1]+g[2]*x1+g[3]*x2+g[4]*x1*x2))),

  #estimating functions for propensity score
  (r/expit(g[5]+g[6]*z+g[7]*x1+g[8]*x2+g[9]*y)-1),
  ((r/expit(g[5]+g[6]*z+g[7]*x1+g[8]*x2+g[9]*y)-1)*z),
  ((r/expit(g[5]+g[6]*z+g[7]*x1+g[8]*x2+g[9]*y)-1)*x1),
  ((r/expit(g[5]+g[6]*z+g[7]*x1+g[8]*x2+g[9]*y)-1)*x2),
  ((r/expit(g[5]+g[6]*z+g[7]*x1+g[8]*x2+g[9]*y)*(y-p.unobs.m.2(g[9],d=c(g[10],g[11]),1,x1,x2,z))
    +p.unobs.m.2(g[9],d=c(g[10],g[11]),1,x1,x2,z))*(z-expit(g[1]+g[2]*x1+g[3]*x2+g[4]*x1*x2))),

  #estimate outcome density parameters
  (r*(y-expit(g[10]+g[11]*x1)) ),
  (r*(y-expit(g[10]+g[11]*x1))*x1 ),

  #estimating function for E(Y)
  ((r/expit(g[5]+g[6]*z+g[7]*x1+g[8]*x2+g[9]*y)
  *(y-p.unobs.m.2(g[9],d=c(g[10],g[11]),1,x1,x2,z))
    +p.unobs.m.2(g[9],d=c(g[10],g[11]),1,x1,x2,z))-g[12])	

)
}

m <- mm.dr.bm(c(pz.x$coef,ans.ipw$par[1:4],dr.sb.bm[i],pcc.m$coef,dr.est.bm[i]))
dr.var.bm[i]<-diag(solve(dM)%*%(m%*%t(m)/n)%*%t(solve(dM))/n)[12]
dr.sbvar.bm[i] <-diag(solve(dM)%*%(m%*%t(m)/n)%*%t(solve(dM))/n)[9]

print(i)
}
\end{Verbatim}
\end{document}